\documentclass[12pt,letterpaper,reqno]{amsart}
\usepackage[english]{babel}
\usepackage{amssymb,graphicx,enumerate}
\usepackage[a4paper,top=2cm,bottom=2cm,left=1.5cm,right=1.5cm,marginparwidth=1.75cm]{geometry}
\newcommand{\assign}{:=}
\newcommand{\nin}{\not\in}
\newcommand{\nobracket}{}
\newcommand{\tmdummy}{$\mbox{}$}

\newcommand{\tmop}[1]{\ensuremath{\operatorname{#1}}}
\newcommand{\tmstrong}[1]{\textbf{#1}}
\newcommand{\tmtextit}[1]{\text{{\itshape{#1}}}}
\newenvironment{enumeratenumeric}{\begin{enumerate}[1.] }{\end{enumerate}}
\newenvironment{enumerateroman}{\begin{enumerate}[i.] }{\end{enumerate}}
\newcounter{nnacknowledgments}

{\theoremstyle{remark}\newtheorem{acknowledgments*}[nnacknowledgments]{Acknowledgments}}
{\theoremstyle{remark}\newtheorem{declarations*}[nnacknowledgments]{Declarations}}

{\theoremstyle{remark}\newtheorem{conflicts of interest*}[nnacknowledgments]{Conflicts of interest}}
\newcounter{nnconvention}

{\theoremstyle{remark}\newtheorem{conventions*}[nnconvention]{Conventions}}
\newtheorem{corollary}{Corollary}
\newtheorem{definition}{Definition}
{\theoremstyle{remark}\newtheorem{example}{Example}}
\newtheorem{lemma}{Lemma}
\newtheorem{proposition}{Proposition}
{\theoremstyle{remark}\newtheorem{remark}{Remark}}
\newtheorem{theorem}{Theorem}
\numberwithin{equation}{section}

\begin{document}

\title{Feynman Graph Integrals on $\mathbb{C}^d$}

\author{Minghao Wang}

\begin{abstract}
  We introduce a type of graph integrals which are holomorphic analogs of
  configuration space integrals. We prove their (ultraviolet) finiteness by
  considering a compactification of the moduli space of graphs with metrics,
  and study their failure to be holomorphic.
\end{abstract}

\maketitle

{\tableofcontents}

\section{Introduction}

Feynman graph integrals play a crucial role not only in physics but also have
numerous important applications in mathematics. One example is the
configuration space integrals of topological field theories developed by S.
Axelrod, I. Singer {\cite{Axelrod1991ChernSimonsPT,axelrod1994chern}}, and M.
Kontsevich {\cite{kontsevich1994feynman}}. These integrals have been shown to play a significant role in advancing the study of knot invariants {\cite{kontsevich1994feynman}},
smooth structures of fiber bundles {\cite{kontsevich1994feynman}}, operads
{\cite{kontsevich1999operads,getzler1994operads}}, deformation quantization
{\cite{kontsevich2003deformation}} and more. One important construction in the
study of configuration space integrals is the construction of compactified
configuration space {\cite{17c23791-d52f-3aa0-a5a9-5abc4d54669e,axelrod1994chern,kontsevich1994feynman}}. For example, the integrand of configuration space
integrals can be extended to a smooth differential form on compactified
configuration space. This extension proves the finiteness of configuration
space integrals.

In this paper, we concentrate on the Feynman graph integrals within holomorphic
field theories, which serve as the holomorphic analogs of configuration space
integrals. For a given graph $\Gamma$ with $n$ vertices and a differential form $\Phi$ with compact support on $(\mathbb{C}^d)^n$, the integrand of our integral is constructed as follows: a kernel differential form is associated with each edge, constructed from the Bochner-Martinelli kernel. The Bochner-Martinelli kernel represents the analytical part of the propagators arising from holomorphic field theories on $(\mathbb{C}^d)^n$. The integrand is then given by the product of all kernel differential forms and $\Phi$. For convenience, we denote the integral by $W_{\Gamma}(\Phi)$. See Section \ref{Feynman graph integrals} for
precise definitions.

Given the singularity of the integrand, an important question arises: Does
$W_{\Gamma} (\Phi)$ yield a finite value (or is convergent in Lebesgue integration theory or other integration theories)? One may try to extend the integrand
to the compactified configuration space. Unfortunately, it is impossible. See
Example \ref{counter example} for a simple counter example. In
{\cite{Li:2011mi}}, S. Li presented a proof of the finiteness in the case when the dimension $d
= 1$, and his proof has been further generalized for one loop graphs for
general $d$ (see
{\cite{Costello2015QuantizationOO,williams2020renormalization}}). In
{\cite{budzik2023feynman}}, K.Budzik, D.Gaiotto, J.Kulp, J.Wu, M.Yu gave a
strategy to prove the finiteness for Laman graphs. In this paper, we will prove it for all
graphs and all dimensions. Motivated by the proof in the topological case, we
realize $W_{\Gamma} (\Phi)$ as an integration of a smooth differential form
over the compactified moduli space of graphs with metrics, which we called
{\tmstrong{compactified Schwinger space}}:

\begin{theorem}[See Theorem \ref{extension theorem1}, Corollary
\ref{distribution theorem}]
  \begin{enumeratenumeric}
    \item $W_{\Gamma} (\Phi)$ can be realized as an integration of a smooth
    differential form over compactified Schwinger space. In particular,
    $W_{\Gamma} (\Phi)$ leads to a finite value.
    
    \item The map $W_{\Gamma} (-) : \Phi \rightarrow W_{\Gamma} (\Phi)$
    defines a distribution on $(\mathbb{C}^d)^n$.
  \end{enumeratenumeric}
\end{theorem}

This theorem generalizes the finiteness result of the Feynman graph integrals for Laman graphs in {\cite{budzik2023feynman}} to general graphs. Although it is argued in {\cite{budzik2023feynman}} that only the Laman graphs appear in the construction of the holomorphic factorization algebra structure, general graph integrals arise in the computation of correlation functions for holomorphic descent observables in holomorphic field theories.

In the study of configuration space integrals on manifolds, the
topological invariance follows from the vanishing results of integrals over
boundaries of compactified configuration space. In the holomorphic setting,
integrals over boundaries of compactified Schwinger space plays a similar
role. We prove the failure for $W_{\Gamma} (-)$ to be holomorphic is given by
integrals over boundaries of compactified Schwinger spaces:

\begin{theorem}[See Theorem \ref{laman anomaly}, Corollary \ref{holomorphicity
failure}]
  $\bar{\partial} W_{\Gamma} (-)$ is a summation of integrals over boundaries
  of compactified Schwinger space labeled by Laman subgraphs of $\Gamma$.
\end{theorem}

The integrals over the boundaries of compactified Schwinger spaces are dominated by contributions from certain special boundaries, denoted by $O_{\Gamma}$. For more details, see Section \ref{Laman graph integral}. Additionally, these integrals are closely connected to anomalies in the Batalin-Vilkovisky formalism of gauge field theories. Further explanations can be found in Remark \ref{QME explanation} and {\cite{minghaoBrian2024}}.

Finally, we present a new
proof of quadratic relations of $O_{\Gamma}$, which appear in
{\cite{budzik2023feynman}}:

\begin{theorem}[K.Budzik, D.Gaiotto, J.Kulp, J.Wu, M.Yu, See
{\cite{budzik2023feynman}} and Theorem \ref{Quadratic relations}]
  We have the following quadratic relations:
  \[ \sum_{\Gamma' \in \tmop{Laman} (\Gamma)} (- 1)^{(d + 1) \sigma (\Gamma',
     \Gamma / \Gamma')} O_{(\Gamma', n |_{\Gamma'} \nobracket)} \circ
     O_{(\Gamma \backslash \Gamma', n |_{\Gamma \backslash \Gamma'}
     \nobracket)} (\Phi) = 0. \]
\end{theorem}

\noindent \textbf{Acknowledgments}
I wound like to thank Keyou Zeng, Chi Zhang, Xujia Chen, Zhengping Gui, Si Li for discussions. I especially thank Maciej Szczesny and Brian Williams for invaluable conversations on this work.
\newline

\noindent \textbf{Conventions}
  {\tmdummy}
  
  \begin{enumerate}
    \item If $S_1, S_2, \ldots S_m$ are sets, $A_1, A_2, \ldots A_m$ are
    finite sets, an element in
    \[ \prod_{i = 1}^m (S_i)^{| A_i |} \]
    will be denoted by $\left( s_{a_1}, s_{a_2} {, \ldots, s_{a_m}} 
    \right)_{a_1 \in A_1, \ldots a_m \in A_m}$, where $s_{a_i}$ is an element
    in $S_i$.
    
    \item Given a manifold $M$, we use $\Omega^{\ast} (M)$ to denote the space
    of differential forms on $M$. The space of distribution-valued
    differential forms is denoted by $\mathcal{D}^{\ast} (M)$. If $M$ is a
    complex manifold, to emphasize their bi-graded nature, we will use
    $\Omega^{\ast, \ast} (M)$ and $\mathcal{D}^{\ast, \ast} (M)$ for
    differential forms and distribution-valued differential forms
    respectively.
    
    \item We use the following Koszul sign rules: If $M, N$ are manifolds
    whose dimensions are $m, n$ respectively, \ $\alpha \in \Omega^i (M),
    \beta \in \Omega^j (N)$
    \begin{enumerateroman}
      \item Sign rule for differential forms:
      \[ \alpha \wedge \beta = (- 1)^{i j} \beta \wedge \alpha . \]
      \item Sign rule for integrations:
      \[ \int_{M \times N} \alpha \wedge \beta = \int_M \int_N \alpha \wedge
         \beta = (- 1)^{n i} \int_M \alpha \int_N \beta . \]
    \end{enumerateroman}
  \end{enumerate}

\section{Feynman graph integrals on $\mathbb{C}^d$.}\label{Feynman graph
integrals}

\subsection{Feynman graph integrals.}

The spacetime we will consider is $d$ dimensional complex affine space
$\mathbb{C}^d$. We will use $\{ z^i \}$ for standard affine coordinates. A
point in $\mathbb{C}^d$ will be denoted as $z$. The notation $d^d
\bar{z} = \underset{i = 1}{\overset{d}{\prod}} d \bar{z}^i$ will always refer
to the top anti-holomorphic form on $\mathbb{C}^d$.

To define the Feynman graph integrals, we first need to construct a distribution-valued differential form:

\begin{definition}
  An {\tmstrong{ordinary propagator}} on $\mathbb{C}^d$ is a distribution-valued differential form
  \[ \tilde{P} (z - w, \bar{z} -
     \bar{w}) \in \mathcal{D}^{\ast, \ast} (\mathbb{C}^d
     \times \mathbb{C}^d), \]
  such that the following equation holds:
  \[ \bar{\partial} (\tilde{P} (z - w, \bar{z} -
     \bar{w})) = \delta (z - w,
     \bar{z} - \bar{w}) d^d (\overline{z -
     w}) . \]
\end{definition}

The choice of propagator is not unique, but we have a standard choice:

\begin{definition}
  The {\tmstrong{Bochner-Martinelli kernel}} is an ordinary propagator given
  by the following formula:
  \[ \frac{(d - 1) !}{\pi^d} \cdot \frac{1}{| z - w |^{2 d}}
     \left( \sum_{i = 1}^d (- 1)^{i - 1} (\overline{z^i - w^i}) \prod_{j \neq
     i} d (\overline{z^j - w^j}) \right) . \]
  We denote it by $\tilde{P}_{0, + \infty} (z - w,
  \bar{z} - \bar{w})$.
\end{definition}

Once we have a propagator, we can define the Feynman graph integrals using the Bochner-Martinelli kernel.

\begin{definition}\label{definition of graphs}
  A directed graph $\Gamma$ consists of the following data:
  \begin{enumeratenumeric}
    \item A set of vertices $\Gamma_0$ and a set of edges $\Gamma_1$.
    
    \item An ordering of all the vertices $\Gamma_0$ and an ordering of all
    the edges $\Gamma_1$.
    
    \item Two maps
    \[ t, h : \Gamma_1 \rightarrow \Gamma_0 \]
    which are the assignments of tail and head to each directed edge.
    
    Furthermore, we say $\Gamma$ is decorated if we have a map
    \[ n : \Gamma_1 \rightarrow (\mathbb{Z}^{\geqslant 0})^d, \text{\quad$e
       \rightarrow (n_{1, e}, n_{2, e}, \ldots, n_{d, e})$} \]
    which associates $d$ non-negative integers to each edge. We use $(\Gamma, n)$
    to denote a decorated graph. If $n$ is the zero map, we will simply write
    $\Gamma$ for $(\Gamma, n)$.
  \end{enumeratenumeric}
\end{definition}

We will use $| \Gamma_0 |$ and $| \Gamma_1 |$ to denote the number of vertices
and edges respectively.

\begin{remark}
    In Definition \ref{definition of graphs}, we require that the graphs be directed because the propagator may not be symmetric, and therefore the Feynman graph integrals, which will be defined below, may depend on the orientation of edges. The purpose of introducing decorations of graphs is to record the order of holomorphic derivatives acting on propagators.
\end{remark}

\begin{definition}
  Given a graph $\Gamma$, the configuration space of $\Gamma$ is
  \[ \tmop{Conf} (\Gamma) = \left\{ \left( z_1,
     z_2, \ldots, z_{| \Gamma_0 |} \right)
     \in (\mathbb{C}^d)^{| \Gamma_0 |} | \nobracket z_i \neq
     z_j  \text{ for any } i, j \right\} . \]
  We can view $z_i$ as the ``position'' of i-th vertex
  of $\Gamma$.
\end{definition}

  Given a decorated graph $(\Gamma, n)$ and a smooth differential form $\Phi
  \in \Omega^{\ast, \ast}_c ((\mathbb{C}^d)^{| \Gamma_0 |})$ with compact
  support, the Feynman graph integral on
  $\mathbb{C}^d$ is formally given by the following expression:
  \begin{eqnarray*}
    &  & W_0^{+ \infty} ((\Gamma, n), \Phi)\\
    & = & (- 1)^{^{\frac{d - 1}{2} | \Gamma_1 | (| \Gamma_1 | - 1) + |
    \Gamma_1 |}} \int_{\tmop{Conf} (\Gamma)} \prod_{e \in \Gamma_1 }
    \partial_{z_{h (e)}}^{n (e)} \tilde{P}_{0, + \infty} (z_{h
    (e)} - z_{t (e)}, \bar{z}_{h (e)} -
    \bar{z}_{t (e)}) \wedge \Phi,
  \end{eqnarray*}
  where $\partial_{z_{h (e)}}^{n (e)} = \partial_{z_{h (e)}^1}^{n_{1, e}}
  \partial_{z_{h (e)}^2}^{n_{2, e}} \ldots \partial_{z_{h (e)}^i}^{n_{i, e}}
  \ldots \partial_{z_{h (e)}^d}^{n_{d, e}}$ is a differential operator with
  constant coefficients, which only involves holomorphic derivatives with
  respect to vertex $h (e)$.

\begin{remark}
   One may wonder why we should consider holomorphic derivatives of propagators. The motivation comes from holomorphic field theories. In a holomorphic field theory, we consider graph integrals constructed according to the following Feynman rules:
\begin{itemize}
    \item For each vertex, we assign a holomorphic local functional,
    \item For each edge, we assign a propagator.
\end{itemize}

The holomorphic local functionals consist of holomorphic derivatives, so the integrand of the graph integrals will be a product of holomorphic derivatives of propagators. See {\cite{williams2020renormalization}} for more details.
\end{remark}

A natural question we should ask is whether the integral $W_0^{+ \infty} ((\Gamma, n),
\Phi)$ is convergent in the sense of Lebesgue integration theory or other (stronger) integration theories. Motivated by the proof strategy in topological field theories (see
{\cite{axelrod1994chern,kontsevich1994feynman}}), one can try to compactify
$\tmop{Conf} (\Gamma)$, so that
\[ \prod_{e \in \Gamma_1 } \partial_{z_{h (e)}}^{n (e)} \tilde{P}_{0, +
   \infty} (z_{h (e)} - z_{t (e)},
   \bar{z}_{h (e)} - \bar{z}_{t (e)}) \]
can be viewed as a smooth differential form over the compactification of
$\tmop{Conf} (\Gamma)$. However, this strategy fails and the Feynman graph integrals may not integrable in the sense of Lebesgue, even in the simplest
example:

\begin{example}
  \label{counter example}Assume $d = 1$, and $\Gamma = (\Gamma, 0)$ is given
  by the following graph:
  \[ 
     \raisebox{-0.24731960585748503\height}{\includegraphics[width=4.462662337662337cm,height=1.0201200314836678cm]{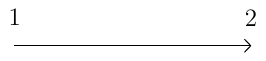}}
  \]
  In this case,
  \begin{eqnarray*}
    \prod_{e \in \Gamma_1 } \partial_{z_{h (e)}}^{n (e)} \tilde{P}_{0, +
    \infty} (z_{h (e)} - z_{t (e)},
    \bar{z}_{h (e)} - \bar{z}_{t (e)}) & = &
    \tilde{P}_{0, + \infty} (z_2 - z_1,
    \bar{z}_2 - \bar{z}_1)\\
    & = & \frac{1}{\pi} \cdot \frac{1}{z_1 - z_2}
  \end{eqnarray*}
  Since $\tilde{P}_{0, + \infty} (z_2 - z_1,
  \bar{z}_2 - \bar{z}_1)$ is an unbounded
  function on $\tmop{Conf} (\Gamma)$, it is not possible to extend it to a
  smooth function over any compactifications of $\tmop{Conf} (\Gamma)$.
\end{example}

In later sections, we will formulate Feynman graph integrals in the sense of heat kernel regularization (see \cite{costellorenormalization}), and prove their convergence. The proof strategy is to use a natural compactification of Schwinger spaces, which will be introduced in the next subsection and Section \ref{compactification'}.

\subsection{Propagator in Schwinger spaces.}

Let
\[
\Delta = \bar{\partial} \circ \bar{\partial}^{\ast} + \bar{\partial}^{\ast} \circ \bar{\partial}
\]
be the standard Laplacian operator in flat space,\footnote{The Laplacian operator used here differs by a factor of $2$ from the usual notation.} where $\bar{\partial}$ is the Dolbeault differential. The adjoint operator $\bar{\partial}^*$ is defined as
\[
\bar{\partial}^{\ast} = - 2 \sum_{i = 1}^d \frac{\partial}{\partial z^i} \iota_{\frac{\partial}{\partial \bar{z}^i}},
\]
where $\iota_{\frac{\partial}{\partial \bar{z}^i}}$ denotes the interior product operator with the vector field $\frac{\partial}{\partial \bar{z}^i}$.
We consider the associated heat kernel
\[ H (t, z, \bar{z}) = \frac{1}{(2 \pi t)^d} e^{-
   \frac{\sum_{i=1}^{d} z^i \overline{z}^i}{2
   t}} d^d \bar{z}, \text{\quad$t > 0$} \]
which is the unique function solving the heat equation
\[ \left\{\begin{array}{l}
     \left( \frac{\partial}{\partial t} + \Delta \right) H (t, z,
     \bar{z}) = 0\\
     H (0, z, \bar{z}) = \delta (z,
     \bar{z}) d^d \bar{z}
   \end{array}\right. \]
\begin{definition}
  The {\tmstrong{regularized ordinary propagator}} $\tilde{P}_{\varepsilon,
  L}$ is defined by the following formula:
  \[ \tilde{P}_{\varepsilon, L} (z - w, \bar{z} -
     \bar{w}) = \int_{\varepsilon}^L \bar{\partial}_z^{\ast}
     H (t, z - w, \bar{z} -
     \bar{w}) d t, \text{\quad$\varepsilon, L > 0$} \]
\end{definition}

The following proposition explains the relation between the regularized ordinary
propagator and the Bochner-Martinelli kernel:

\begin{proposition}
  The following equality holds:
  \[ \tilde{P}_{0, + \infty} (z - w, \bar{z} -
     \bar{w}) = \lim_{\underset{L \rightarrow +
     \infty}{\varepsilon \rightarrow 0}} \tilde{P}_{\varepsilon, L} (z -
     w, \bar{z} - \bar{w}) \]
\end{proposition}

\begin{proof}
  The proof follows by straightforward calculation
\end{proof}

\begin{definition}
  Given a decorated graph $(\Gamma, n)$ and a smooth differential form $\Phi
  \in \Omega^{\ast, \ast}_c ((\mathbb{C}^d)^{| \Gamma_0 |})$ with compact
  support, we define the following {\tmstrong{regularized Feynman graph
  integral}}:
  \begin{eqnarray*}
    &  & W_{\varepsilon}^L ((\Gamma, n), \Phi)\\
    & = & (- 1)^{^{\frac{d - 1}{2} | \Gamma_1 | (| \Gamma_1 | - 1) + |
    \Gamma_1 |}} \int_{(\mathbb{C}^d)^{| \Gamma_0 |}} \prod_{e \in \Gamma_1 }
    \partial_{z_{h (e)}}^{n (e)} \tilde{P}_{\varepsilon, L}
    (z_{h (e)} - z_{t (e)},
    \bar{z}_{h (e)} - \bar{z}_{t (e)})
    \wedge \Phi .
  \end{eqnarray*}
  The {\tmstrong{Feynman graph
  integral}} is defined by the following limit (if it exists):
  \[
  W_{0}^{+\infty} ((\Gamma, n), \Phi)=\lim_{\underset{L \rightarrow +
     \infty}{\varepsilon \rightarrow 0}}W_{\varepsilon}^{L} ((\Gamma, n), \Phi)  .
  \]
\end{definition}
We will also use $W_0^L ((\Gamma, n), \Phi)$ to denote the limit (if it exists)
  \[ \lim_{\varepsilon \rightarrow 0} W_{\varepsilon}^L ((\Gamma, n), \Phi) . \]

The following proposition shows that the existence of $W_{0}^{+\infty} ((\Gamma, n), \Phi)$ is a consequence of the existence of $W_{0}^L ((\Gamma, n), \Phi)$.

\begin{proposition}
  \label{finiteness transfer}Let $L>0$. If $W_0^L ((\Gamma, n), \Phi)$ exists for any $(\Gamma, n)$, and any compactly supported differential form $\Phi \in \Omega^{\ast, \ast}_c ((\mathbb{C}^d)^{|
  \Gamma_0 |})$, then $W_0^{+ \infty} ((\Gamma, n), \Phi)$ exists for any
  $(\Gamma, n)$ and any $\Phi \in \Omega^{\ast, \ast}_c ((\mathbb{C}^d)^{|
  \Gamma_0 |})$.
\end{proposition}

\begin{proof}
  Let $0<\varepsilon<L<L'$. Given any $(\Gamma, n)$ and $\Phi \in \Omega^{\ast, \ast}_c
  ((\mathbb{C}^d)^{| \Gamma_0 |})$, we have
  \begin{eqnarray}\label{L-infty equation}
    &  & (- 1)^{^{\frac{d - 1}{2} | \Gamma_1 | (| \Gamma_1 | - 1) + |
    \Gamma_1 |}} W_{\varepsilon}^{L'} ((\Gamma, n), \Phi)\nonumber\\
    & = & \int_{\tmop{Conf} (\Gamma)} \prod_{e \in \Gamma_1 } \partial_{z_{h
    (e)}}^{n (e)} \tilde{P}_{\varepsilon, L'} (z_{h (e)} -
    z_{t (e)}, \bar{z}_{h (e)} -
    \bar{z}_{t (e)}) \wedge \Phi\nonumber\\
    & = &\int_{\tmop{Conf} (\Gamma)} \prod_{e \in \Gamma_1 } \partial_{z_{h
    (e)}}^{n (e)} (\tilde{P}_{\varepsilon, L} + \tilde{P}_{L, L'})
    (z_{h (e)} - z_{t (e)},
    \bar{z}_{h (e)} - \bar{z}_{t (e)})
    \wedge \Phi\nonumber\\
    & = &  \sum_{\Gamma' \subseteq \Gamma} (\pm) W_{\varepsilon}^L \left( (\Gamma', n | 
    \nobracket_{\Gamma'}), \prod_{e \nin \Gamma'_1 } \partial_{z_{h (e)}}^{n
    (e)} \tilde{P}_{L, L'} (z_{h (e)} -
    z_{t (e)}, \bar{z}_{h (e)} -
    \bar{z}_{t (e)}) \wedge \Phi \right).
  \end{eqnarray}
  We notice that $\lim_{L'\rightarrow+\infty}\tilde{P}_{L, L'}=\tilde{P}_{L,+\infty}$ is a smooth
  differential form, so 
  \[
  \lim_{L'\rightarrow+\infty}\prod_{e \nin \Gamma'_1 } \partial_{z_{h (e)}}^{n
    (e)} \tilde{P}_{L, L'} (z_{h (e)} -
    z_{t (e)}, \bar{z}_{h (e)} -
    \bar{z}_{t (e)}) \wedge \Phi
  \]
  exists in $\Omega^{\ast, \ast}_c
  ((\mathbb{C}^d)^{| \Gamma_0 |})$. It is clear that the iterated limit $\lim_{\varepsilon\rightarrow0}\lim_{L'\rightarrow+\infty}$ of the right hand side of $(\ref{L-infty equation})$ exists. Now, we show that the double limit $\lim_{\underset{L' \rightarrow +
     \infty}{\varepsilon \rightarrow 0}}$ of the right hand side of $(\ref{L-infty equation})$ exists.

     We notice that each term of the right hand side of $(\ref{L-infty equation})$ has the form $T_{\varepsilon}(\phi_{L'})$, where
     
     \begin{enumerate}
         \item $T_{\varepsilon}\in \mathcal{D}^{\ast, \ast} ((\mathbb{C}^d)^{| \Gamma_0 |})$ and $T_{0}(\phi)=\lim_{\varepsilon\rightarrow0}T_{\varepsilon}(\phi)$ exists for any $\phi\in\Omega^{\ast, \ast}_c
  ((\mathbb{C}^d)^{| \Gamma_0 |})$. 
  \item $\phi_{L'}\in\Omega^{\ast, \ast}_c
  ((\mathbb{C}^d)^{| \Gamma_0 |})$ is convergent to $\phi_{+\infty}\in\Omega^{\ast, \ast}_c
  ((\mathbb{C}^d)^{| \Gamma_0 |})$ when $L'\rightarrow+\infty$.
     \end{enumerate}
       We prove $\lim_{\underset{L' \rightarrow +
     \infty}{\varepsilon \rightarrow 0}}T_{\varepsilon}(\phi_{L'})$ exists.
     
     Let $\{(\varepsilon_{i},L'_i)\}_{i=1}^{\infty}$ be a sequence, such that $0<\varepsilon_{i}<L<L'_{i}$ and $\lim_{i\rightarrow\infty}(\varepsilon_{i},L'_i)=(0,+\infty)$. By our assumption, $\{T_{\varepsilon_i}\}_{i=1}^{\infty}$ is a sequence of distributions, such that the pointwise limit of $T_{\varepsilon_i}$ exists. By Banach–Steinhaus theorem for $\Omega^{\ast, \ast}_c
  ((\mathbb{C}^d)^{| \Gamma_0 |})$, $\{T_{\varepsilon_i}\}_{i=1}^{\infty}$ is equicontinuous. So, for any $\delta>0$, we can find a neighborhood $U\subset\Omega^{\ast, \ast}_c
  ((\mathbb{C}^d)^{| \Gamma_0 |})$ of $\phi_{+\infty}$, such that 
  \[
  |T_{\varepsilon_i}(\phi)-T_{\varepsilon_i}(\phi_{+\infty})|<\frac{1}{2}\delta \qquad\text{for all }i\in \mathbb{N} \text{ and }\phi\in U.
  \]
  We can choose some $N\in \mathbb{N}$, such that when $i>N$, we have $\phi_{L'_{i}}\in U$ and 
  \[|T_{\varepsilon_i}(\phi_{+\infty})-T_{0}(\phi_{+\infty})|<\frac{1}{2}\delta,\]
  so when $i>N$, we have 
  \[
  |T_{\varepsilon_i}(\phi_{L'_{i}})-T_0(\phi_{+\infty})|<\delta.
  \]
  This proves the existence of $\lim_{\underset{L' \rightarrow +
     \infty}{\varepsilon \rightarrow 0}}T_{\varepsilon}(\phi_{L'})$. Now, we see that the double limit $\lim_{\underset{L' \rightarrow +
     \infty}{\varepsilon \rightarrow 0}}$ of the right hand side of $(\ref{L-infty equation})$ exists, so $W_{0}^{+\infty} ((\Gamma, n), \Phi)$ exists.
\end{proof}

\begin{remark}
    Readers who have a physical background may wonder whether Proposition \ref{finiteness transfer} implies the absence of infrared divergence. The answer is no: the differential form $\Phi$ plays the role of an infrared cutoff because it has a compact support.
\end{remark}

Proposition \ref{finiteness transfer} shows that we only need to show the existence of $W_0^L
((\Gamma, n), \Phi)$. Now, we want to rephrase $W_0^L ((\Gamma, n),
\Phi)$ in terms of the following propagator in Schwinger spaces:

\begin{definition}
  Given $t > 0$, the {\tmstrong{propagator in Schwinger space}} $P_t$ is
  defined by the following formula:
  \[ P_t (z - w, \bar{z} -
     \bar{w}) = - d t \wedge \bar{\partial}_z^{\ast} H (t,
     z - w, \bar{z} - \bar{w}) +
     H (t, z - w, \bar{z} -
     \bar{w}) . \]
  We will simply call $P_t$ the propagator if there is no ambiguity.
\end{definition}

One important reason to introduce this propagator is the following lemma:

\begin{lemma}
  \label{y}Let $y = \frac{\bar{z} -
  \bar{w}}{2t}$, we have
  \[ P_t (z - w, \bar{z} -
     \bar{w}) = \frac{1}{\pi^d} e^{- (z - w)
     \cdot y} d^d y, \]
     where $(z - w)
     \cdot y=\sum_{i=1}^{d}(z^i-w^i)y^i$.
\end{lemma}

\begin{proof}
  We perform the following calculations:
  \begin{eqnarray*}
    &  & \frac{1}{\pi^d} e^{- (z - w) \cdot y} d^d y\\
    & = & \frac{1}{\pi^d} e^{- (z - w) \cdot y} \prod_{i =
    1}^d d y^i\\
    & = & \frac{1}{\pi^d} e^{- (z - w) \cdot y} \left(
    \frac{1}{(2 t)^d} \prod_{i = 1}^d d \left( 
    \overline{z}^i - \overline{w}^i \right)
    \right.\\
    & - & \left. \frac{1}{2^d (t)^{d + 1}} \sum_{i = 1}^d (- 1)^{i - 1}
    \left( \overline{z}^i - \overline{w}^i
    \right) dt \prod_{j \neq i}^d d \left( 
    \overline{z}^j - \overline{w}^j \right)
    \right)\\
    & = & H (t, z - w, \bar{z} -
    \bar{w}) - d t \wedge \bar{\partial}_z^{\ast} H (t,
    z - w, \bar{z} - \bar{w})
  \end{eqnarray*}
  
\end{proof}

The above propagator can be found in {\cite{budzik2023feynman}}. With this
propagator, we have rephrase $W_{\varepsilon}^L ((\Gamma, n), \Phi)$ in the
following way:

\begin{proposition}
  Given a decorated graph $(\Gamma, n)$ and $\Phi \in \Omega^{\ast, \ast}_c
  ((\mathbb{C}^d)^{| \Gamma_0 |})$, we have the following equality:
  \[ W_{\varepsilon}^L ((\Gamma, n), \Phi) = \int_{(\mathbb{C}^d)^{| \Gamma_0
     |} \times [\varepsilon, L]^{| \Gamma_1 |}} \prod_{e \in \Gamma_1 }
     \partial_{z_{h (e)}}^{n (e)} P_{t_e} (z_{h (e)} -
     z_{t (e)}, \bar{z}_{h (e)}-
     \bar{z}_{t (e)}) \wedge \Phi, \]
  where $t_e$ is the parameter associated with each edge $e \in \Gamma_1$.
\end{proposition}

\begin{proof}
  Note in the integrand
  \[ \prod_{e \in \Gamma_1 } \partial_{z_{h (e)}}^{n (e)} P_{t_e}
     (z_{h (e)} - z_{t (e)},
     \bar{z}_{h (e)} - \bar{z}_{t (e)})
     \wedge \Phi, \]
  $| \Gamma_1 |$ propagators appear and each provides at most $1$
  $dt_e$ component. Furthermore, the number of propagators equal to the
  dimension of $[\varepsilon, L]^{| \Gamma_1 |}$, so we have the following
  equality:
  \begin{eqnarray*}
    &  & \int_{(\mathbb{C}^d)^{| \Gamma_0 |} \times [\varepsilon, L]^{|
    \Gamma_1 |}} \prod_{e \in \Gamma_1 } \partial_{z_{h (e)}}^{n (e)} P_{t_e}
    (z_{h (e)} - z_{t (e)},
    \bar{z}_{h (e)} - \bar{z}_{t (e)})
    \wedge \Phi\\
    & = & (- 1)^{^{| \Gamma_1 |}} \int_{(\mathbb{C}^d)^{| \Gamma_0 |} \times
    [\varepsilon, L]^{| \Gamma_1 |}} \prod_{e \in \Gamma_1 } \partial_{z_{h
    (e)}}^{n (e)} (d t_e \wedge \bar{\partial}^{\ast}_{z_{h (e)}} H (t_e,
    z_{h (e)} - z_{t (e)},
    \bar{z}_{h (e)} - \bar{z}_{t (e)}))
    \wedge \Phi\\
    & = & (- 1)^{^{\frac{d - 1}{2} | \Gamma_1 | (| \Gamma_1 | - 1) + |
    \Gamma_1 |}} \cdot\\
    &  & \int_{(\mathbb{C}^d)^{| \Gamma_0 |}} \int_{[\varepsilon, L]^{|
    \Gamma_1 |}} \prod_{e \in \Gamma_1 } d t \prod_{e \in \Gamma_1 }
    \partial_{z_{h (e)}}^{n (e)} \bar{\partial}^{\ast}_{z_{h (e)}} H (t_e,
    z_{h (e)} - z_{t (e)},
    \bar{z}_{h (e)} - \bar{z}_{t (e)})
    \wedge \Phi\\
    & = & (- 1)^{^{\frac{d - 1}{2} | \Gamma_1 | (| \Gamma_1 | - 1) + |
    \Gamma_1 |}} \int_{(\mathbb{C}^d)^{| \Gamma_0 |}} \prod_{e \in \Gamma_1 }
    \partial_{z_{h (e)}}^{n (e)} \tilde{P}_{\varepsilon, L}
    (z_{h (e)} - z_{t (e)},
    \bar{z}_{h (e)} - \bar{z}_{t (e)})
    \wedge \Phi\\
    & = & W_{\varepsilon}^L ((\Gamma, n), \Phi)
  \end{eqnarray*}
  In the second-to-last step, we used the fact that only the top differential forms contributes the
  integral over $[\varepsilon, L]^{| \Gamma_1 |}$.
\end{proof}

\begin{definition}
  Given a directed graph $\Gamma$, and $L > 0$, the Schwinger space is defined
  by $(0, L]^{| \Gamma_1 |}$. The orientation is given by the following
  formula:
  \[ \int_{(0, L]^{| \Gamma_1 |}} \prod_{e \in \Gamma_1} d t_e = L^{| \Gamma_1
     |} . \]
\end{definition}

\begin{proposition}
  Given a decorated graph $(\Gamma, n)$ and $\Phi \in \Omega^{\ast, \ast}_c
  ((\mathbb{C}^d)^{| \Gamma_0 |})$, we denote the integrand
  \[ \prod_{e \in \Gamma_1 } \partial_{z_{h (e)}}^{n (e)} P_{t_e}
     (z_{h (e)} - z_{t (e)},
     \bar{z}_{h (e)} - \bar{z}_{t (e)})
     \wedge \Phi \]
  by $\tilde{W} ((\Gamma, n), \Phi)$. Then
  \[ \int_{(\mathbb{C}^d)^{| \Gamma_0 |}} \tilde{W} ((\Gamma, n), \Phi) \in
     \Omega^{\ast} ((0, L]^{| \Gamma_1 |})  \text{for any } L > 0. \]
\end{proposition}

\begin{proof}
  This can be proved by dominated convergence theorem easily. We will prove a
  stronger version later.
\end{proof}

Our strategy to prove the finiteness of $W_{\varepsilon}^L ((\Gamma, n),
\Phi)$ can be described by the following steps:
\begin{enumeratenumeric}
  \item Find a compactification $\widetilde{[0, L]^{| \Gamma_1 |}}$ of
  Schwinger space $(0, L]^{| \Gamma_1 |}$, which is a manifold with corners.
  
  \item Prove $\int_{(\mathbb{C}^d)^{| \Gamma_0 |}} \tilde{W} ((\Gamma, n),
  \Phi)$ can be extended to a differential form over $\widetilde{[0, L]^{|
  \Gamma_1 |}}$,
  
  \item Since $\widetilde{[0, L]^{| \Gamma_1 |}}$ is compact, the integral
  over $\widetilde{[0, L]^{| \Gamma_1 |}}$ is finite.
\end{enumeratenumeric}
Let us end this subsection by stating two useful properties of the propagator
in Schwinger space.

\begin{lemma}
  \label{closeness of propagator}{\tmdummy}
  
  \begin{enumeratenumeric}
    \item Let $d_t$ be the de-Rham differential over Schwinger space, then
    \[ (d_t + \bar{\partial}) P_t (z - w, \bar{z}
       - \bar{w}) = 0. \]
    \item The following equality holds:
    \[ \left( t \iota_{\frac{\partial}{\partial t}} + \sum_{i = 1}^d \bar{z}^i
       \iota_{\frac{\partial}{\partial \bar{z}^i}} + \sum_{i = 1}^d \bar{w}^i
       \iota_{\frac{\partial}{\partial \bar{w}^i}} \right) P_t (z - w,
       \bar{z} - \bar{w}) = 0 \]
  \end{enumeratenumeric}
\end{lemma}

\begin{proof}
  These can be shown by direct computations.
\end{proof}

\subsection{Feynman graph integrals in Schwinger spaces.}

In this subsection, we will rephrase our Feynman graph integral in a more
convenient coordinate system.

We first note the following facts:
\begin{enumeratenumeric}
  \item If a decorated graph $(\Gamma, n)$ has two connected components
  $(\Gamma', n | \nobracket_{\Gamma'})$ and $(\Gamma'', n |
  \nobracket_{\Gamma''})$, we have
  \[ W_{\varepsilon}^L ((\Gamma, n), \Phi' \cdot \Phi'') = \pm
     W_{\varepsilon}^L ((\Gamma', n | \nobracket_{\Gamma'}), \Phi')
     W_{\varepsilon}^L ((\Gamma'', n | \nobracket_{\Gamma''}), \Phi''), \]
  where $\Phi' \in \Omega^{\ast, \ast}_c ((\mathbb{C}^d)^{| \Gamma'_0 |})$,
  $\Phi'' \in \Omega^{\ast, \ast}_c ((\mathbb{C}^d)^{| \Gamma''_0 |})$.
  
  \item If a decorated graph $(\Gamma, n)$ contains a self-loop, then
  \[ W_{\varepsilon}^L ((\Gamma, n), \Phi) = 0. \]
\end{enumeratenumeric}
\ Therefore, we assume $\Gamma$ is connected without self-loops in the
subsequent content.

Given $(\Gamma, n)$, we introduce the following coordinate transformation:
\[ \left\{\begin{array}{l}
     z_i = w_i +  w_{|
     \Gamma_0 |} \text{\qquad} 1 \leqslant i \leqslant | \Gamma_0 | - 1\\
     z_{| \Gamma_0 |} = w_{| \Gamma_0 |}
   \end{array}\right. \]
Here we have used the ordering of vertices of $\Gamma$.

We introduce the following incidence matrix:

\begin{definition}
  Given a connected graph $\Gamma$ without self-loops, we have the following
  incidence matrix:
  \[ \rho^e_i = \left\{\begin{array}{l}
       1 \text{\qquad if } h (e) = i\\
       - 1 \text{\quad if } t (e) = i\\
       0 \text{\qquad otherwise}
     \end{array}\right. \]
  where $i \in \Gamma_0$, $e \in \Gamma_1$.
\end{definition}

Using lemma \ref{y} and coordinates $\{ w_i^j \}$, the integrand of the Feynman
graph integral becomes:
\begin{equation}
  \tilde{W} ((\Gamma, n), \Phi) = \frac{1}{\pi^{d | \Gamma_1 |}} e^{-
  \overset{| \Gamma_0 | - 1}{\underset{i = 1}{\sum}} \underset{e \in
  \Gamma_1}{\sum} \rho_i^e w_i \cdot y_e}
  \prod_{e \in \Gamma_1} \left( \prod_{1 \leqslant i \leqslant d}
  (y^i_e)^{n_{i, e}} \right) d^d y_e \wedge \Phi, \label{integrand}
\end{equation}
where
\[ y_e = \sum_{i = 1}^{| \Gamma_0 | - 1} \frac{1}{2 t_e}
   \rho^e_i \bar{w}_i . \]

The following proposition shows the power of $\left( \ref{integrand} \right)$:

\begin{proposition}
  If there exists a connected subgraph $\Gamma' \subseteq \Gamma$, such that
  \[ d | \Gamma'_0 | < (d - 1) | \Gamma_1' | + d + 1, \]
  then $\tilde{W} ((\Gamma, n), \Phi) = 0$.\label{vanishing result1}
\end{proposition}

\begin{proof}
  Note there is a factor $\tilde{W} ((\Gamma', n |  \nobracket_{\Gamma'}), 1)$
  in $\tilde{W} ((\Gamma, n), \Phi)$, we only need to prove $\tilde{W}
  ((\Gamma', n |  \nobracket_{\Gamma'}), 1) = 0$.
  
  Let's consider the following map $g_{\Gamma'}$:
  \[ g_{\Gamma'} : (w_i, \bar{w}_i, t_e) \in
     (\mathbb{C}^d)^{| \Gamma_0' |} \times (0, + \infty)^{| \Gamma_1' |}
     \rightarrow (y_e) \in (\mathbb{C}^d)^{| \Gamma_1' |}, \]
  since $g_{\Gamma'}$ is an anti-holomorphic map, we have the following
  anti-holomorphic tangent map:
  \[ T^{\tmop{anti}}_{(w_i, \bar{w}_i, t_e)}
     g : T^{\tmop{anti}}_{(w_i, \bar{w}_i)}
     (\mathbb{C}^d)^{| \Gamma_0' |} \oplus T_{(t_e)} \rightarrow
     T^{\tmop{holo}}_{(y_e)} (\mathbb{C}^d)^{| \Gamma_1' |} .
  \]
  We notice that $\tilde{W} ((\Gamma', n |  \nobracket_{\Gamma'}), 1)$
  contains a factor $g_{\Gamma'}^{\ast} \left( \prod_{e \in \Gamma'_1} d^d y_e
  \right)$. Since $\prod_{e \in \Gamma'_1} d^d y_e$ is a top holomorphic form,
  \[ g_{\Gamma'}^{\ast} {\left( \prod_{e \in \Gamma'_1} d^d y_e
     \right)_{(w_i, \bar{w}_i, t_e)}}  = 0
  \]
  if $T^{\tmop{anti}}_{(w_i, \bar{w}_i,
  t_e)} g$ is not surjective. We further notice that the following vectors
  belong to the kernel of $T^{\tmop{anti}}_{(w_i,
  \bar{w}_i, t_e)} g$:
  \[ \sum_{i = 1}^{| \Gamma_0' |} \partial_{w_i^j}, \text{\quad} \sum_{i =
     1}^{| \Gamma_1 |} t_e \partial_{t_e} + \sum_{\underset{1 \leqslant k
     \leqslant d}{1 \leqslant i \leqslant | \Gamma_0 |}} \bar{w}_i^k
     \partial_{\bar{w}_i^k}, \text{\quad} 1 \leqslant j \leqslant d, \]
  so the dimension of the image of $T^{\tmop{anti}}_{(w_i,
  \bar{w}_i, t_e)} g$ is bounded:
  \[ \dim \left( \tmop{Im} \left( T^{\tmop{anti}}_{(w_i,
     \bar{w}_i, t_e)} g \right) \right) \leqslant d |
     \Gamma_0' | + | \Gamma_1' | - d - 1 < d | \Gamma_1' | . \]
  Note the dimension of $T^{\tmop{holo}}_{(y_e)}
  (\mathbb{C}^d)^{| \Gamma_1' |}$ is $d | \Gamma_1' |$,
  $T^{\tmop{anti}}_{(w_i, \bar{w}_i, t_e)}
  g$ cannot be surjective. We conclude that
  \[ \tilde{W} ((\Gamma, n), \Phi) = 0 \]
\end{proof}

\begin{remark}
    The proof of Proposition \ref{vanishing result1} is motivated by the following fact: Let $M$ and $N$ be smooth manifolds, and let $\alpha$ be a top form on $N$. Suppose $f: M \to N$ is a smooth map such that the tangent map of $f$ is not surjective at any point. Then the pullback of $\alpha$ along $f$ is zero. It remains unclear whether this result has a physical interpretation.
\end{remark}

We can rephrase the Feynman graph integral by:
\begin{eqnarray}
  &  & W_{\varepsilon}^L ((\Gamma, n), \Phi) \nonumber\\
  & = & \frac{1}{\pi^{d | \Gamma_1 |}} \int_{\mathbb{C}^d_{w_{| \Gamma_0 |}}}
  \int_{(\mathbb{C}^d)^{| \Gamma_0 | - 1} \times [\varepsilon, L]^{| \Gamma_1
  |}} e^{- \overset{| \Gamma_0 | - 1}{\underset{i = 1}{\sum}} \underset{e \in
  \Gamma_1}{\sum} \rho_i^e w_i \cdot y_e}\nonumber\\
  & &
  \prod_{e \in \Gamma_1} \left( \prod_{1 \leqslant i \leqslant d}
  (y^i_e)^{n_{i, e}} \right) d^d y_e \wedge \Phi,  \label{integral}
\end{eqnarray}
where $\int_{\mathbb{C}^d_{w_{| \Gamma_0 |}}}$ is the notation denotes the
integral over $\{ w_{| \Gamma_0 |}^i \}$. Since
\[ e^{- \overset{| \Gamma_0 | - 1}{\underset{i = 1}{\sum}} \underset{e \in
   \Gamma_1}{\sum} \rho_i^e w_i \cdot y_e}
   \prod_{e \in \Gamma_1} \left( \prod_{1 \leqslant i \leqslant d}
   (y^i_e)^{n_{i, e}} \right) d^d y_e \]
is a smooth differential form in terms of $\{ w_i^j, y_e^j \}$, we can imagine
that if we can find a compactification of Schwinger space and proper
coordinates, such that $\{ w_i^j, y_e^j \}$ are smooth functions, we may
extend $\int_{\mathbb{C}^d_{w_{| \Gamma_0 |}}} \int_{(\mathbb{C}^d)^{|
\Gamma_0 | - 1}} \tilde{W} ((\Gamma, n), \Phi)$ to the compactification space.
This will be achieved in section \ref{compactification}.

In the remaining part of this subsection, we will present some concepts and
facts from graph theory, which will be used in later sections.

\begin{definition}
  Given a connected graph $\Gamma$, a tree $T \subseteq \Gamma$ is said to be
  a spanning tree if every vertex of $\Gamma$ lies in $T$. We denote the set
  of all spanning tree by $\tmop{Tree} (\Gamma)$.
\end{definition}

\begin{definition}
  Given a connected graph $\Gamma$ and two disjoint subsets of vertices $V_1,
  \text{ } V_2 \subseteq \Gamma_0$, we define $\tmop{Cut} (\Gamma ; V_1, V_2)$
  to be the set of subsets $C \subseteq \Gamma_1$ satisfying the following
  properties:
  \begin{enumeratenumeric}
    \item The removing of edges in $C$ from $\Gamma$ divides $\Gamma$ into
    exactly two connected trees, which we denoted by $\Gamma' (C), \text{ }
    \Gamma'' (C) $, such that $V_1 \subseteq \Gamma_0' (C), \text{ } V_2
    \subseteq \Gamma_0'' (C)$.
    
    \item $C$ doesn't contain any proper subset satisfying property 1.
  \end{enumeratenumeric}
\end{definition}

\begin{remark}
    Intuitively, a cut is a partition of a given graph into two connected components, each of which is a tree. An element of $\tmop{Cut}(\Gamma; V_1, V_2)$ consists of the edges that need to be removed to obtain such a cut.
\end{remark}

\begin{definition}
  Given a connected graph $\Gamma$ without self-loops, and a function mapping
  each $e \in \Gamma_1$ to $t_e \in (0, + \infty)$, we define the weighted
  laplacian of $\Gamma$ by the following formula:
  \[ M_{\Gamma} (t)_{\tmop{ij}} = \sum_{e \in \Gamma_1} \rho^e_i \frac{1}{t_e}
     \rho_j^e, \text{\quad} 1 \leqslant i, j \leqslant | \Gamma_0 | - 1 \]
\end{definition}

\begin{remark}
  The integrand $\tilde{W} ((\Gamma, n), \Phi)$ contains an exponential factor,
  and its exponent is given by:
  \[ - \frac{1}{2} \sum_{1 \leqslant i, j \leqslant | \Gamma_0 | - 1}
     w_i M_{\Gamma} (t)_{i j} \cdot
     \bar{w}_j, \]
  the appearance of weighted laplacian of graphs in the Feynman graph integrals is
  surprising.
\end{remark}

The following facts will be used to extend $\int_{\mathbb{C}^d_{w_{| \Gamma_0
|}}} \int_{(\mathbb{C}^d)^{| \Gamma_0 | - 1}} \tilde{W} ((\Gamma, n), \Phi)$
to compactified Schwinger space.

\begin{theorem}[Kirchhoff]
  Given a connected graph $\Gamma$ without self-loops, the determinant of
  weighted laplacian is given by the following formula:
  \[ \det M_{\Gamma} (t) = \frac{\underset{T \in \tmop{Tree} (\Gamma)}{\sum}
     \underset{e \nin T}{\prod} t_e}{\underset{e \in \Gamma_1}{\prod} t_e} \]
  
\end{theorem}

\begin{corollary}
  The inverse of $M_{\Gamma} (t)$ is given by the following
  formula:\label{Minverse}
  \[ (M_{\Gamma} (t)^{- 1})^{i j} = \frac{1}{\underset{T \in \tmop{Tree}
     (\Gamma)}{\sum} \underset{e \nin T}{\prod} t_e} \cdot \left( \sum_{C \in
     \tmop{Cut} (\Gamma ; \{ i, j \}, \{ | \Gamma_0 | \})} \prod_{e \in C} t_e
     \right) \]
\end{corollary}

\begin{corollary}
  \label{boundness}We have the following equality:
  \begin{eqnarray*}
    &  & \frac{1}{t_e} \sum_{i = 1}^{| \Gamma_0 | - 1} \rho^e_i (M_{\Gamma}
    (t)^{- 1})^{i j}\\
    & = & \frac{1}{\underset{T \in \tmop{Tree} (\Gamma)}{\sum} \underset{e
    \nin T}{\prod} t_e} \left( \sum_{C \in \tmop{Cut} (\Gamma ; \{ j, h (e)
    \}, \{ | \Gamma_0 |, t (e) \})} \frac{\prod_{e' \in C} t_{e'}}{t_e}
    \right.\\
    & - & \left. \sum_{C \in \tmop{Cut} (\Gamma ; \{ j, t (e) \}, \{ |
    \Gamma_0 |, h (e) \})} \frac{\prod_{e' \in C} t_{e'}}{t_e} \right)
  \end{eqnarray*}
  In particular, every term of the numerator also appears in the denominator,
  so we have
  \[ \left| \frac{1}{t_e} \sum_{i = 1}^{| \Gamma_0 | - 1} \rho^e_i (M_{\Gamma}
     (t)^{- 1})^{i j} \right| \leqslant 2. \]
  We use $(d_{\Gamma}^{- 1})^{e j}$ to denote $\frac{1}{t_e} \sum_{i = 1}^{|
  \Gamma_0 | - 1} \rho^e_i (M_{\Gamma} (t)^{- 1})^{i j}$.
\end{corollary}

\begin{proof}
  See {\cite[Appendix B]{Li:2011mi}}.
\end{proof}

\begin{remark}
  If we view graphs as discrete spaces, the incidence matrix can be viewed as
  de Rham differential. Then $\frac{1}{t_e} \sum_{i = 1}^{| \Gamma_0 | - 1}
  \rho^e_i (M_{\Gamma} (t)^{- 1})^{i j}$ can be viewed as the Green's function
  of de Rham differential on a graph. This might explains the importance of
  $\frac{1}{t_e} \sum_{i = 1}^{| \Gamma_0 | - 1} \rho^e_i (M_{\Gamma} (t)^{-
  1})^{i j}$.
\end{remark}

\section{Compactification of Schwinger spaces and finiteness of the Feynman graph
integrals.}\label{compactification}

\subsection{Compactification of Schwinger spaces.}\label{compactification'}

Assume $\Gamma$ is a connected directed graph without self-loops, $L > 0$ is a
positive real number. Let $S \subseteq \Gamma_1$ be a subset of $\Gamma_1$, we
define the following submanifold with corners of Schwinger space:
\[ \Delta_S = \left\{ (t_1, t_2, \ldots, t_{| \Gamma_1 |}) \in [0, L]^{|
   \Gamma_1 |} | \nobracket t_e = 0 \quad \tmop{if} e \in S \right\} . \]
The compactification of Schwinger space is obtained by iterated real blow ups
of $[0, L]^{| \Gamma_1 |}$ along $\Delta_S$ for all $S \subseteq \Gamma_1$ in
a certain order(see {\cite{epub47792,Ammann2019ACO}}). To avoid getting into
technical details of real blow ups of manifolds with corners, we will use
another simpler definition. Instead, we present a typical example of real blow
up, which will be helpful to understand our construction:

\begin{example}
  Let $S = \{ 1, 2, \ldots, k \} \subseteq \Gamma_1$, the real blow up of $[0,
  + \infty)^{| \Gamma_1 |}$ along $\Delta_S$ is the following manifold (with
  corners):
  \[ [[0, + \infty)^{| \Gamma_1 |} : \Delta_S] \assign \left\{ (\rho, \xi_1,
     \ldots, \xi_k, t_{k + 1}, \ldots t_{| \Gamma_1 |}) \in [0, + \infty)^{|
     \Gamma_1 | + 1} \left| \sum_{i = 1}^k \xi_i^2 = 1 \right. \right\} . \]
  We have a natural map from $[[0, + \infty)^{| \Gamma_1 |} : \Delta_S]$ to
  $[0, + \infty)^{| \Gamma_1 |}$:
  \[ (\rho, \xi_1, \ldots, \xi_k, t_{k + 1}, \ldots t_{| \Gamma_1 |})
     \rightarrow (t_1 = \rho \xi_1, \ldots, t_k = \rho \xi_k, t_{k + 1},
     \ldots, t_{| \Gamma_1 |}) . \]
  We also have a natural inclusion map from $(0, + \infty)^{| \Gamma_1 |}$ to
  $[[0, + \infty)^{| \Gamma_1 |} : \Delta_S]$:
  \[ (t_1, \ldots, t_{| \Gamma_1 |}) \rightarrow \left( \rho = \sqrt{\sum_{i =
     1}^k t_i^2}, \xi_1 = \frac{t_1}{\sqrt{\sum_{i = 1}^k t_i^2}}, \ldots,
     \xi_k = \frac{t_k}{\sqrt{\sum_{i = 1}^k t_i^2}}, t_{k + 1}, \ldots, t_{|
     \Gamma_1 |} \right) . \]
  For us, the most important property is that $\frac{t_i}{\sqrt{\sum_{i = 1}^k
  t_i^2}}$ can be extended to a smooth function $\xi_i$ on $[[0, + \infty)^{|
  \Gamma_1 |} : \Delta_S]$.
\end{example}

Let's consider the following natural inclusion map:
\[ i : (0, + \infty)^{| \Gamma_1 |} \rightarrow \prod_{S \subseteq \Gamma_1}
   [[0, + \infty)^{| \Gamma_1 |} : \Delta_S] \]
\begin{definition}
  We call the closure of the image of $(0, L]^{| \Gamma_1 |}$ under $i$ the
  {\tmstrong{compactified Schwinger space}} of $\Gamma$. We denote it by
  $\widetilde{[0, L]^{| \Gamma_1 |}}$.
\end{definition}

\begin{remark}
  Sometimes, with a slight abuse of terminology, we also call the closure of the image under $i$ of $ (0, +\infty)^{|\Gamma_1|} $ the compactified Schwinger space, even though it is not compact. We will use $\widetilde{[0, + \infty)^{|
  \Gamma_1 |}}$ to denote it.
\end{remark}

\begin{proposition}
  $\widetilde{[0, L]^{| \Gamma_1 |}}$ is a compact manifold with corners.
\end{proposition}

\begin{proof}
  See {\cite{Ammann2019ACO}}.
\end{proof}

To obtain a more concrete description of $\widetilde{[0, L]^{| \Gamma_1 |}}$,
we introduce a useful concept called corners of compactified Schwinger spaces.

For $\Gamma_1 = S_0 \supseteq S_1 \supsetneq S_2 \supsetneq \cdots \supsetneq
S_m \supsetneq S_{m + 1} = \varnothing$, we define a submanifold
\[ C_{S_1, S_2, \ldots, S_m} \subseteq [0, + \infty)^m \times (0, + \infty)^{|
   S_0 | - | S_1 |} \times \prod_{i = 1}^m (0, + \infty)^{| S_i | - | S_{i +
   1} |}, \]
which is given by a set of equations. The coordinates are as follows: we use $\{ \rho_i \} (1 \leqslant i \leqslant m)$ to denote the coordinates of
the first product component $[0, + \infty)^m$, use \[
\{ t_e \} (e \in S_0
\backslash S_1)\]
 to denote the coordinates of the product component $(0, +
\infty)^{| S_0 | - | S_1 |}$, and use $\{ \xi_e^i \} \left( 1 \leqslant i
\leqslant m, e \in S_i {\backslash S_{i + 1}}  \right)$ to denote the coordinates of
the product component $(0, + \infty)^{| S_{i - 1} | - | S_i |}$. Then the submanifold is defined by the following set of equations:
\[ \left\{\begin{array}{l}
     \underset{e \in S_m}{\sum} (\xi_e^m)^2 = 1\\
     \underset{e \in S_{m - 1} {\backslash S_m} }{\sum} (\xi_e^{m - 1})^2 +
     \underset{{e \in S_m} }{\sum} (\rho_m \xi_e^m)^2 = 1\\
     \underset{e \in S_{m - 2} {\backslash S_{m - 1}} }{\sum} (\xi_e^{m -
     2})^2 + \underset{e \in S_{m - 1} {\backslash S_m} }{\sum} (\rho_{m - 1}
     \xi_e^{m - 1})^2 + \underset{{e \in S_m} }{\sum} (\rho_{m - 1} \rho_m
     \xi_e^m)^2 = 1\\
     \ldots\\
     \underset{i = 1}{\overset{k}{\sum}} \left( \underset{e \in S_{m - k + i}
     {\backslash S_{m - k + i + 1}} }{\sum} \left( \left( \underset{j =
     1}{\overset{i - 1}{\prod}} \rho_{m - k + j} \right) \xi_e^{m - k + i}
     \right)^2 \right) = 1\\
     \ldots\\
     \underset{i = 1}{\overset{m}{\sum}} \left( \underset{e \in S_{m - k + i}
     {\backslash S_{m - k + i + 1}} }{\sum} \left( \left( \underset{j =
     1}{\overset{i - 1}{\prod}} \rho_{m - k + j} \right) \xi_e^{m - k + i}
     \right)^2 \right) = 1
   \end{array}\right. \]
\begin{proposition}
  \label{covered by corners}The following statements are true:
  \begin{enumeratenumeric}
    \item There is a natural inclusion map from $C_{S_1, S_2, \ldots, S_m}$ to
    $\widetilde{[0, + \infty)^{| \Gamma_1 |}}$.
    
    \item $\widetilde{[0, + \infty)^{| \Gamma_1 |}}$ is covered by the union
    of all $C_{S_1, S_2, \ldots, S_m}$:
    \[ \widetilde{[0, + \infty)^{| \Gamma_1 |}} = \bigcup_{\Gamma_1 = S_0
       \supseteq S_1 \supsetneq \cdots \supsetneq S_m \supsetneq S_{m + 1} =
       \varnothing} C_{S_1, S_2, \ldots, S_m} . \]
    \item There is a natural action of $\mathbb{R}_+ = (0, + \infty)$ on
    $\widetilde{[0, + \infty)^{| \Gamma_1 |}}$, which has the following form
    on $(0, + \infty)^{| \Gamma_1 |} \subseteq \widetilde{[0, + \infty)^{|
    \Gamma_1 |}}$:
    \[ \lambda \cdot (t_e)_{e \in \Gamma_1} = (\lambda t_e)_{e \in \Gamma_1}
       \quad \lambda \in (0, + \infty) . \]
  \end{enumeratenumeric}
\end{proposition}

\begin{proof}
  The proofs are based on elementary reasoning.
\end{proof}

\begin{definition}
  For $\Gamma_1 = S_0 \supseteq S_1 \supsetneq S_2 \supsetneq \cdots
  \supsetneq S_m \supsetneq S_{m + 1} = \varnothing$, let ${{\Gamma'}^i} $ be
  the subgraph generated by $S_i (1 \leqslant i \leqslant m)$. We call
  $C_{S_1, S_2, \ldots, S_m} \subseteq \widetilde{[0, + \infty)^{| \Gamma_1
  |}}$ the {\tmstrong{corner of compactified Schwinger space}} corresponds to
  ${\Gamma'}^1 {, \Gamma'}^2, \ldots,  {\Gamma'}^m \rightarrow 0$.
\end{definition}

In later sections, we will use $\partial_{m} C_{S_1, S_2, \ldots, S_m}$ to denote the following set:
\[ \left\{ p \in C_{S_1, S_2, \ldots, S_m} | \nobracket \rho_i (p) = 0 \text{
   for } 1 \leqslant i \leqslant m \right\} . \]

\begin{remark}
  Note $\widetilde{[0, + \infty)^{| \Gamma_1 |}}$ is a manifold with corners.
  In particular, it is a stratified space. Its codimension m strata is
  \[ \bigcup_{\Gamma_1 = S_0 \supseteq S_1 \supsetneq \cdots \supsetneq S_m
     \supsetneq S_{m + 1} = \varnothing} \partial_{m} C_{S_1, S_2, \ldots, S_m} .
  \]
  We will denote the closure of $\partial C_{\Gamma_{1}}$ in $\widetilde{[0, + \infty)^{| \Gamma_1 |}}$ by $\widetilde{(0, +
     \infty)^{| \Gamma_1' |}} /\mathbb{R}^+$
  
\end{remark}

By proposition \ref{covered by corners}, to extend a differential form on \
$(0, + \infty)^{| \Gamma_1 |}$ to $\widetilde{[0, + \infty)^{| \Gamma_1 |}}$,
we only need to extend it to a smooth differential form on $C_{S_1, S_2,
\ldots, S_m}$.

\subsection{Two technical lemmas}

In this subsection, we prove that both $(M_{\Gamma} (t)^{- 1})^{i j}$ and
$(d_{\Gamma}^{- 1})^{e j}$ can be extended to $\widetilde{[0, + \infty)^{|
\Gamma_1 |}}$. We will see the extension of \ $\int_{\mathbb{C}^d_{w^{|
\Gamma_0 |}}} \int_{(\mathbb{C}^d)^{| \Gamma_0 | - 1}} \tilde{W} ((\Gamma, n),
\Phi)$ follows in the next subsection.

\begin{lemma}
  \label{key}Given a connected graph $\Gamma$ without self-loops, $\Gamma'
  \subseteq \Gamma$ is a subgraph, we use $h_1 (\Gamma)$ and $h_1 (\Gamma')$
  to denote the dimension of the 1st order homology groups of $\Gamma$ and
  $\Gamma'$ respectively. Then on the corner $C_{\Gamma'_1}$ of compactified
  Schwinger space, we have
  \[ \underset{T \in \tmop{Tree} (\Gamma)}{\sum} \underset{e \nin T}{\prod}
     t_e = \sum^{h_1 (\Gamma)}_{i = h_1 (\Gamma')} \rho^i P_i, \]
  where $\{ \rho \} \cup \{ \xi_{e'} \}_{e' \in \Gamma'_1} \cup \{ t_e \}_{e
  \in \Gamma_1 \backslash \Gamma'_1}$ are the functions on $C_{\Gamma'_1}$
  which have been defined in section \ref{compactification'}. $P_i$ is a
  degree $i$ homogeneus polynomial with variables $\{ \xi_{e'} \}_{e' \in
  \Gamma'_1} \cup \{ t_e \}_{e \in \Gamma_1 \backslash \Gamma'_1}$.
  Furthermore, $P_{h_1 (\Gamma')} \neq 0$.
\end{lemma}

\begin{proof}
  We note that $t_{e'} = \rho \xi_{e'}$ for $e' \in \Gamma'_1$, we only need
  to count the number of edges in $\Gamma' \backslash T$, for each spanning
  tree $T \in \tmop{Tree}$.
  
  We notice that $\Gamma' \cap T$ should be a tree in $\Gamma'$, the number of
  edges in $\Gamma' \backslash T$ should be larger than $h_1 (\Gamma')$. Since
  $T$ is a spanning tree of $\Gamma$, the number of edges in $\Gamma
  \backslash T$ equals to $h_1 (\Gamma)$.
  
  It remains to show $P_{h_1 (\Gamma')} \neq 0$. This follows from the following
  fact: if $T'$ is a spanning tree in $\Gamma'$, there exists a spanning tree $T$ in
  $\Gamma$, such that
  \[ T' = T \cap \Gamma . \]
  So there is at least one term in $\underset{T \in \tmop{Tree}
  (\Gamma)}{\sum} \underset{e \nin T}{\prod} t_e$ contributes $P_{h_1
  (\Gamma')}$.
\end{proof}

\begin{lemma}
  \label{extended functions}Given a connected graph $\Gamma$ without
  self-loops, The following functions can be extended to smooth functions on
  $\widetilde{[0, + \infty)^{| \Gamma_1 |}}$:
  \begin{enumeratenumeric}
    \item $(M_{\Gamma} (t)^{- 1})^{i j}$ for $1 \leqslant i, j \leqslant |
    \Gamma_0 | - 1$.
    
    \item $(d_{\Gamma}^{- 1})^{e j}$ for $e \in \Gamma_1, \text{ } 1 \leqslant
    j \leqslant | \Gamma_0 | - 1$.
  \end{enumeratenumeric}
\end{lemma}

\begin{proof}
  To avoid unnecessary complexity in the notation, we will only prove that $(M_{\Gamma} (t)^{- 1})^{i j}$ and $(d_{\Gamma}^{- 1})^{e j}$
  can be extended to $C_{\Gamma'_1}$, where $\Gamma'_1$ is a subgraph of
  $\Gamma$. The general situation can be resolved through a straightforward generalization.
  \begin{enumeratenumeric}
    \item By Corollary \ref{Minverse}, we have
    \[ (M_{\Gamma} (t)^{- 1})^{i j} = \frac{1}{\underset{T \in \tmop{Tree}
       (\Gamma)}{\sum} \underset{e \nin T}{\prod} t_e} \cdot \left( \sum_{C
       \in \tmop{Cut} (\Gamma ; \{ i, j \}, \{ | \Gamma_0 | \})} \prod_{e \in
       C} t_e \right) . \]
    If we rewrite the numerator in terms of $\{ \rho \} \cup \{ \xi_{e'}
    \}_{e' \in \Gamma'_1} \cup \{ t_e \}_{e \in \Gamma_1 \backslash
    \Gamma'_1}$, we have
    \[ \sum_{C \in \tmop{Cut} (\Gamma ; \{ i, j \}, \{ | \Gamma_0 | \})}
       \prod_{e \in C} t_e = \sum^{h_1 (\Gamma) + 1}_{i = h_1 (\Gamma')}
       \rho^i \tilde{P}_i, \]
    where $\tilde{P}_i$ is a degree $i$ homogeneus polynomial with variables
    $\{ \xi_{e'} \}_{e' \in \Gamma'_1} \cup \{ t_e \}_{e \in \Gamma_1
    \backslash \Gamma'_1}$. By lemma \ref{key}, we have
    \[ (M_{\Gamma} (t)^{- 1})^{i j} = \frac{1}{\underset{i = h_1
       (\Gamma')}{\overset{h_1 (\Gamma)}{\sum}} \rho^{i - h_1 (\Gamma')} P_i,}
       \left( \sum^{h_1 (\Gamma) + 1}_{i = h_1 (\Gamma')} \rho^{i - h_1
       (\Gamma')} \tilde{P}_i \right), \]
    note the denominator is not zero, it is a smooth function on
    $C_{\Gamma'_1}$.
    
    \item By Corollary \ref{boundness}, we have
    \begin{eqnarray*}
      &  & (d_{\Gamma}^{- 1})^{e j}\\
      & = & \frac{1}{\underset{T \in \tmop{Tree} (\Gamma)}{\sum} \underset{e
      \nin T}{\prod} t_e} \left( \sum_{C \in \tmop{Cut} (\Gamma ; \{ j, h (e)
      \}, \{ | \Gamma_0 |, t (e) \})} \frac{\prod_{e' \in C} t_{e'}}{t_e}
      \right.\\
      & - & \left. \sum_{C \in \tmop{Cut} (\Gamma ; \{ j, t (e) \}, \{ |
      \Gamma_0 |, h (e) \})} \frac{\prod_{e' \in C} t_{e'}}{t_e} \right) .
    \end{eqnarray*}
    Similarly, we have
    \begin{eqnarray*}
      &  & \sum_{C \in \tmop{Cut} (\Gamma ; \{ j, h (e) \}, \{ | \Gamma_0 |,
      t (e) \})} \frac{\prod_{e' \in C} t_{e'}}{t_e} - \sum_{C \in \tmop{Cut}
      (\Gamma ; \{ j, t (e) \}, \{ | \Gamma_0 |, h (e) \})} \frac{\prod_{e'
      \in C} t_{e'}}{t_e}\\
      & = & \sum^{h_1 (\Gamma)}_{i = h_1 (\Gamma')} \rho^i
      \widetilde{\tilde{P}}_i,
    \end{eqnarray*}
    where $\widetilde{\tilde{P}}_i$ is a degree $i$ homogeneus polynomial with
    variables $\{ \xi_{e'} \}_{e' \in \Gamma'_1} \cup \{ t_e \}_{e \in
    \Gamma_1 \backslash \Gamma'_1}$. By lemma \ref{key}, we have
    \[ (d_{\Gamma}^{- 1})^{e j} = \frac{1}{\underset{i = h_1
       (\Gamma')}{\overset{h_1 (\Gamma)}{\sum}} \rho^{i - h_1 (\Gamma')} P_i,}
       \left( \sum^{h_1 (\Gamma)}_{i = h_1 (\Gamma')} \rho^{i - h_1 (\Gamma')}
       \widetilde{\tilde{P}}_i \right) . \]

    Since the denominator is not zero, $(d_{\Gamma}^{- 1})^{e j}$ is smooth on
    $C_{\Gamma'_1}$.
  \end{enumeratenumeric}
\end{proof}

\subsection{Finiteness of Feynman graph integrals.}

In this subsection, we will prove the main theorem. The idea is find a
``coordinate transformation'' of $(\mathbb{C}^d)^{| \Gamma_0 |} \times (0, +
\infty)^{| \Gamma_1 |}$, such that the integrand $\tilde{W} ((\Gamma, n),
\Phi)$ can be expressed in terms of $(M_{\Gamma} (t)^{- 1})^{i j}$ ,
$(d_{\Gamma}^{- 1})^{e j}$ and their de Rham differentials.

Intuitively, the ``coordinate transformation'' is given by the following
``map'':
\[ (\tilde{w}_i, \overline{\tilde{w}}_i,
   t_e)_{1 \leqslant i \leqslant | \Gamma_0 | - 1, e \in \Gamma_1} \rightarrow
   \left( w_i = \tilde{w}_i,
   \bar{w}_i = \sum_{j = 1}^{| \Gamma_0 | - 1} (M_{\Gamma}
   (t)^{- 1})^{i j} \overline{\tilde{w}}_j, t_e \right)_{1
   \leqslant i \leqslant | \Gamma_0 | - 1, e \in \Gamma_1}, \label{coordinate
   transformation} \]
if we use ``coordinates'' $\{ \tilde{w}_i \}_{i \in \Gamma_0}
\cup \{ \overline{\tilde{w}}_i \}_{i \in \Gamma_0} \cup \{
t_e \}_{e \in \Gamma_1}$, we will have
\begin{equation}
  y_e = \sum_{i = 1}^{| \Gamma_0 | - 1} \frac{1}{2 t_e}
  \rho^e_i \bar{w}_i = \sum_{i = 1}^{| \Gamma_0 | - 1}
  \sum_{j = 1}^{| \Gamma_0 | - 1} \frac{1}{2 t_e} \rho^e_i (M_{\Gamma} (t)^{-
  1})^{i j} \overline{\tilde{w}}_j = \sum_{j = 1}^{| \Gamma_0
  | - 1} \frac{1}{2} (d_{\Gamma}^{- 1})^{e j}
  \overline{\tilde{w}}_j, \label{111}
\end{equation}
so the integrand
\begin{eqnarray}
  &  & \tilde{W} ((\Gamma, n), \Phi) \nonumber\\
  & = & \frac{1}{\pi^{d | \Gamma_1 |}} e^{- \overset{| \Gamma_0 | -
  1}{\underset{i = 1}{\sum}} \underset{e \in \Gamma_1}{\sum}
  \rho_i^e \tilde{w}_i \cdot y_e} \prod_{e
  \in \Gamma_1} \left( \prod_{1 \leqslant i \leqslant d} (y^i_e)^{n_{i, e}}
  \right) d^d y_e \nonumber\\
  & \wedge & \Phi \left( \tilde{w}_i, \sum_{j = 1}^{|
  \Gamma_0 | - 1} (M_{\Gamma} (t)^{- 1})^{i j}
  \overline{\tilde{w}}_j \right)  \label{222}
\end{eqnarray}
can be extended to $(\mathbb{C}^d)^{| \Gamma_0 |} \times \widetilde{[0, +
\infty)^{| \Gamma_1 |}}$. The only problem is that $\left( \ref{coordinate
transformation} \right)$ is not a map: $\tilde{w}_i \text{
and } \overline{\tilde{w}}_i$ are not independent of each
other. Therefore, we need to use a trick to change the integration domain
itself.

\begin{definition}
  Let $k = \left( k_1, \ldots k_{|
  \Gamma_0 | - 1} \right) \in (\mathbb{C}^d)^{| \Gamma_0 | - 1}$, we call
  $\Phi_{k} \in \Omega^{\ast, \ast} ((\mathbb{C}^d)^{| \Gamma_0 |})$ a
  simple harmonic wave form with momenta $\left( k_1, \ldots
  k_{| \Gamma_0 | - 1} \right)$, if it can be written as
  \[ \Phi_{k} = e^{\overset{| \Gamma_0 | - 1}{\underset{i = 1}{\sum}}
     w_i \cdot k_i - \overset{| \Gamma_0 | -
     1}{\underset{i = 1}{\sum}} \bar{w}_i \cdot
     \bar{k}_i} \omega, \]
  where $\omega$ is a differential form such that the coefficients are
  functions of $\{ w^i_{| \Gamma_0 |} \}$, i.e. $\omega$ is constant along
  $(\mathbb{C}^d)^{| \Gamma_0 | - 1}$.
\end{definition}

We first notice that $\tilde{W} ((\Gamma, n), \Phi_{k})$ can be viewed
as a differential form on $\mathbb{C}^d_{w_{| \Gamma_0 |}} \times
\mathbb{C}^{| \Gamma_0 | - 1} \times \overline{\mathbb{C}^{| \Gamma_0 | - 1}}
\times \widetilde{[0, + \infty)^{| \Gamma_1 |}} $\footnote{This is because
$\tilde{W} ((\Gamma, n), \Phi_{k})$ is an analytic form. We will use the
same notation for convenience.}, where $\overline{\mathbb{C}^{| \Gamma_0 | -
1}}$ is the complex affine space with the holomorphic structure given by the
anti-holomorphic structure on $\mathbb{C}^{| \Gamma_0 | - 1}$. A point in
$\overline{\mathbb{C}^{| \Gamma_0 | - 1}}$ will be denoted by
$(\bar{w}_i)_{1 \leqslant i \leqslant | \Gamma_0 | - 1}$.
Further more, $\tilde{W} ((\Gamma, n), \Phi_{k})$ is a holomorphic form
along $\mathbb{C}^{| \Gamma_0 | - 1} \times \overline{\mathbb{C}^{| \Gamma_0 |
- 1}}$.

Let $\Delta, \Delta'$ be embedings from
\[ \mathbb{C}^d_{w_{| \Gamma_0 |}} \times \mathbb{C}^{| \Gamma_0 | - 1} \times
   (0, + \infty)^{| \Gamma_1 |} \]
to
\[ \mathbb{C}^d_{w_{| \Gamma_0 |}} \times \mathbb{C}^{| \Gamma_0 | - 1}
   \times \overline{\mathbb{C}^{| \Gamma_0 | - 1}} \times (0, + \infty)^{|
   \Gamma_1 |}, \]
they are given by the following formulas:
\begin{eqnarray*}
  &  & \Delta \left( \left( \tilde{w}_{| \Gamma_0 | - 1},
  \tilde{w}_i, t_e \right)_{1 \leqslant i \leqslant |
  \Gamma_0 | - 1, e \in \Gamma_1} \right)\\
  & = & \left( w_{| \Gamma_0 | - 1} =
  \tilde{w}_{| \Gamma_0 | - 1}, w_i =
  \tilde{w}_i, \bar{w}_i =
  \overline{\tilde{w}}_i, t_e \right)_{1 \leqslant i
  \leqslant | \Gamma_0 | - 1, e \in \Gamma_1},
\end{eqnarray*}
\begin{eqnarray*}
  &  & \Delta' \left( \left( \tilde{w}_{| \Gamma_0 | - 1},
  \tilde{w}_i, t_e \right)_{1 \leqslant i \leqslant |
  \Gamma_0 | - 1, e \in \Gamma_1} \right)\\
  & = & \left( w_{| \Gamma_0 | - 1} =
  \tilde{w}_{| \Gamma_0 | - 1}, w_i =
  \tilde{w}_i, \bar{w}_i = \sum_{j = 1}^{|
  \Gamma_0 | - 1} (M_{\Gamma} (t)^{- 1})^{i j}
  \overline{\tilde{w}}_j, t_e \right)_{1 \leqslant i
  \leqslant | \Gamma_0 | - 1, e \in \Gamma_1} .
\end{eqnarray*}
\begin{proposition}
  \label{change integration cycle}We have the following equality:
  \[ \int_{(\mathbb{C}^d)^{| \Gamma_0 | - 1}} \Delta^{\ast} \tilde{W}
     ((\Gamma, n), \Phi_{k}) = \int_{(\mathbb{C}^d)^{| \Gamma_0 | - 1}}
     {\Delta'}^{\ast} \tilde{W} ((\Gamma, n), \Phi_{k}) . \]
\end{proposition}

\begin{proof}
  Note the space of all positive symmetric matrices is contractible, we can
  find a continuous family $A_t$, such that $A_0 = \tmop{Id}$, $A_1 =
  M_{\Gamma} (t)$. Then we define
  \[ \Delta_t : \mathbb{C}^d_{w^{| \Gamma_0 |}} \times \mathbb{C}^{| \Gamma_0
     | - 1} \times (0, + \infty)^{| \Gamma_1 |} \rightarrow \mathbb{C}^d_{w^{|
     \Gamma_0 |}} \times \mathbb{C}^{| \Gamma_0 | - 1} \times
     \overline{\mathbb{C}^{| \Gamma_0 | - 1}} \times (0, + \infty)^{| \Gamma_1
     |} \]
  by
  \[ \begin{array}{lll}
       &  & \Delta_t \left( \left( \tilde{w}_{| \Gamma_0 | -
       1}, \tilde{w}_i, t_e \right)_{1 \leqslant i \leqslant
       | \Gamma_0 | - 1, e \in \Gamma_1} \right)\\
       & = & \left( w_{| \Gamma_0 | - 1} =
       \tilde{w}_{| \Gamma_0 | - 1}, w_i =
       \tilde{w}_i, \bar{w}_i = \sum_{j =
       1}^{| \Gamma_0 | - 1} (A_t^{- 1})^{i j}
       \overline{\tilde{w}}_j, t_e \right)_{1 \leqslant i
       \leqslant | \Gamma_0 | - 1, e \in \Gamma_1} .
     \end{array} \]

  $\Delta_t$ defines a homotopy between $\Delta$ and $\Delta'$. We notice that
  $\Delta_t^{\ast} \tilde{W} ((\Gamma, n), \Phi_{k})$ decays faster than
  any polynomials when $(\tilde{w}_i)_{1 \leqslant i
  \leqslant | \Gamma_0 | - 1}$ goes to infinity, and $\Delta_t^{\ast}
  \tilde{W} ((\Gamma, n), \Phi_{k})$ is a top holomorphic
  form (otherwise, $\int_{(\mathbb{C}^d)^{| \Gamma_0 | - 1}} \Delta_t^{\ast}
  \tilde{W} ((\Gamma, n), \Phi_{k}) = 0$), we have
  \[ \int_{(\mathbb{C}^d)^{| \Gamma_0 | - 1}} \Delta^{\ast} \tilde{W}
     ((\Gamma, n), \Phi_{k}) = \int_{(\mathbb{C}^d)^{| \Gamma_0 | - 1}}
     {\Delta'}^{\ast} \tilde{W} ((\Gamma, n), \Phi_{k}) \]
  by Stokes theorem.
\end{proof}

By Lemma \ref{extended functions}, $\left( \ref{111} \right)$ and $\left(
\ref{222} \right)$, ${\Delta'}^{\ast} \tilde{W} ((\Gamma, n), \Phi_{k})$
can be extended to $\mathbb{C}^d_{w_{| \Gamma_0 |}} \times \mathbb{C}^{|
\Gamma_0 | - 1} \times \widetilde{[0, + \infty)^{| \Gamma_1 |}}$. By integrate
out $\mathbb{C}^{| \Gamma_0 | - 1}$, we have the following result:

\begin{lemma}
  The following statements are true:\label{extension lemma}
  \begin{enumeratenumeric}
    \item $\int_{(\mathbb{C}^d)^{| \Gamma_0 | - 1}} \tilde{W} ((\Gamma, n),
    \Phi_{k}) = \int_{(\mathbb{C}^d)^{| \Gamma_0 | - 1}} \Delta^{\ast}
    \tilde{W} ((\Gamma, n), \Phi_{k})$ can be extended to a smooth
    differential form on $\mathbb{C}^d_{w_{| \Gamma_0 |}} \times
    \widetilde{[0, L]^{| \Gamma_1 |}}$.
    
    \item Assume $D$ is any order $| D |$ differential operator on on
    $\mathbb{C}^d_{w_{| \Gamma_0 |}} \times \widetilde{[0, L]^{| \Gamma_1
    |}}$, we have the following inequality:
    \[ \left| D \int_{(\mathbb{C}^d)^{| \Gamma_0 | - 1}} {\Delta'}^{\ast}
       \tilde{W} ((\Gamma, n), \Phi_{k}) \right| \leqslant C'' \left( 1
       + \sum_{i = 1}^{| \Gamma_0 | - 1} | k_i |
       \right)^{\underset{e \in \Gamma_1, 1 \leqslant i \leqslant
       d}{\overset{}{\sum}} n_{i, e} + | \Gamma_1 | + 2 | D |}, \]
    where $C''$ is some constant.
  \end{enumeratenumeric}
\end{lemma}

\begin{proof}
  
  \begin{enumeratenumeric}
    \item Since ${\Delta'}^{\ast} \tilde{W} ((\Gamma, n), \Phi_{k})$ can
    be extended to $\mathbb{C}^d_{w_{| \Gamma_0 |}} \times \mathbb{C}^{|
    \Gamma_0 | - 1} \times \widetilde{[0, + \infty)^{| \Gamma_1 |}}$, if we
    can prove \[
    \int_{(\mathbb{C}^d)^{| \Gamma_0 | - 1}} {\Delta'}^{\ast}
    \tilde{W} ((\Gamma, n), \Phi_{k})\]
     is smooth, the first part of this
    lemma follows from proposition \ref{change integration cycle}. We will use
    dominated convergence theorem to prove the smoothness.
    
    First, we note ${\Delta'}^{\ast} \tilde{W} ((\Gamma, n), \Phi_{k})$
    can be rewrited in the following form:
    \begin{eqnarray*}
      &  & {\Delta'}^{\ast} \tilde{W} ((\Gamma, n), \Phi_{k})\\
      & = & \frac{1}{\pi^{d | \Gamma_1 |}} e^{- \frac{1}{2} \underset{i =
      1}{\overset{| \Gamma_0 | - 1}{\sum}} \tilde{w}_i \cdot
      \overline{\tilde{w}}_i + \overset{| \Gamma_0 | -
      1}{\underset{i = 1}{\sum}} \tilde{w}_i \cdot
      k_i - \overset{| \Gamma_0 | - 1}{\underset{i =
      1}{\sum}} \underset{j = 1}{\overset{| \Gamma_0 | - 1}{\sum}}
      \overline{\tilde{w}}_i \cdot (M_{\Gamma} (t)^{- 1})^{i
      j} \bar{k}_j} P
      (\overline{\tilde{w}}_i) \omega,
    \end{eqnarray*}
    where $\omega$ is a differential form that is constant along
    $(\mathbb{C}^d)^{| \Gamma_0 | - 1}$. $P
    (\overline{\tilde{w}}_i)$ is a polynomial with
    coefficients in terms of $(M_{\Gamma} (t)^{- 1})^{i j}$, $(d_{\Gamma}^{-
    1})^{e j}$ and their derivatives. The degree of $P
    (\overline{\tilde{w}}_i)$ is less than $\underset{e \in
    \Gamma_1, 1 \leqslant i \leqslant d}{\overset{}{\sum}} n_{i, e} + |
    \Gamma_1 |$.
    
    Since $\widetilde{[0, L]^{| \Gamma_1 |}}$ is compact, the values of
    $(M_{\Gamma} (t)^{- 1})^{i j}$ and $(d_{\Gamma}^{- 1})^{e j}$ are bounded
    by some constant $C$. We also notice the following inequalities:
    \[ \left\{\begin{array}{l}
         \left| \overset{| \Gamma_0 | - 1}{\underset{i = 1}{\sum}}
         \tilde{w}_i \cdot k_i \right|
         \leqslant \frac{\varepsilon^2}{2} \underset{i = 1}{\overset{|
         \Gamma_0 | - 1}{\sum}} \tilde{w}_i \cdot
         \overline{\tilde{w}}_i + \frac{1}{2 \varepsilon^2}
         \underset{i = 1}{\overset{| \Gamma_0 | - 1}{\sum}}
         k_i \cdot \bar{k}_i\\
         \left| \overset{| \Gamma_0 | - 1}{\underset{i = 1}{\sum}} \underset{j
         = 1}{\overset{| \Gamma_0 | - 1}{\sum}}
         \overline{\tilde{w}}_i \cdot (M_{\Gamma} (t)^{-
         1})^{i j} \bar{k}_j \right| \leqslant
         \frac{\varepsilon^2}{2} \underset{i = 1}{\overset{| \Gamma_0 | -
         1}{\sum}} \tilde{w}_i \cdot
         \overline{\tilde{w}}_i + \frac{(| \Gamma_0 | - 1)^2
         C^2}{2 \varepsilon^2} \underset{i = 1}{\overset{| \Gamma_0 | -
         1}{\sum}} k_i \cdot \bar{k}_i
       \end{array}\right., \]
    where $\varepsilon > 0$ is an arbitrary small number. So
    \begin{eqnarray}
      &  & \left| D \left( {\Delta'}^{\ast} \tilde{W} ((\Gamma, n),
      \Phi_{k}) \right) \right| \nonumber\\
      & \leqslant & C' e^{- \frac{1 - 2 \varepsilon^2}{2} \underset{i =
      1}{\overset{| \Gamma_0 | - 1}{\sum}} \tilde{w}_i \cdot
      \overline{\tilde{w}}_i}{\left( 1 + \sum_{1 = 1}^{|
      \Gamma_0 | - 1} | \tilde{w}_i | \right)^{\underset{e
      \in \Gamma_1, 1 \leqslant i \leqslant d}{\overset{}{\sum}} n_{i, e} + |
      \Gamma_1 | + d (| \Gamma_0 | - 1) + | D |}} ,  \label{inequality1}
    \end{eqnarray}
    where $D$ is any order $| D |$ differential operator on
    $\mathbb{C}^d_{w_{| \Gamma_0 |}} \times \widetilde{[0, L]^{| \Gamma_1
    |}}$, $C'$ is some constant which may depend on $k$. Since the right
    hand side of \ref{inequality1} is absolute integrable for small
    $\varepsilon$ and independent of $\mathbb{C}^d_{w_{| \Gamma_0 |}} \times
    \widetilde{[0, L]^{| \Gamma_1 |}}$, $\int_{(\mathbb{C}^d)^{| \Gamma_0 | -
    1}} {\Delta'}^{\ast} \tilde{W} ((\Gamma, n), \Phi_{k})$ is smooth
    over $\mathbb{C}^d_{w_{| \Gamma_0 |}} \times \widetilde{[0, L]^{| \Gamma_1
    |}}$ by dominated convergence theorem.
    
    \item We note $\int_{(\mathbb{C}^d)^{| \Gamma_0 | - 1}} {\Delta'}^{\ast}
    \tilde{W} ((\Gamma, n), \Phi_{k})$ is a Gaussian type integral, so
    \begin{eqnarray*}
      &  & \frac{\pi^{d | \Gamma_1 |}}{(2 \pi)^{d (| \Gamma_0 | - 1)}}
      \int_{(\mathbb{C}^d)^{| \Gamma_0 | - 1}} {\Delta'}^{\ast} \tilde{W}
      ((\Gamma, n), \Phi_{k})\\
      & = & e^{\frac{1}{2} \underset{i = 1}{\overset{| \Gamma_0 | - 1}{\sum}}
      \overset{d}{\underset{j = 1}{\sum}} (\partial_{\tilde{w}_i^j}
      \partial_{\overline{\tilde{w}}_i^j} + k_i^j
      \partial_{\overline{\tilde{w}}_i^j} - \bar{k}_i^j
      \partial_{\tilde{w}_i^j}) - \frac{1}{2} \overset{| \Gamma_0 | -
      1}{\underset{i = 1}{\sum}} \underset{j = 1}{\overset{| \Gamma_0 | -
      1}{\sum}} k_i \cdot (M_{\Gamma} (t)^{- 1})^{i j}
      \bar{k}_j} P (\overline{\tilde{w}}_i)
      \omega \left|_{\tilde{w}_i = 0} \right.\\
      & = & e^{\frac{1}{2} \underset{i = 1}{\overset{| \Gamma_0 | - 1}{\sum}}
      \overset{d}{\underset{j = 1}{\sum}} (\partial_{\tilde{w}_i^j}
      \partial_{\overline{\tilde{w}}_i^j} + k_i^j
      \partial_{\overline{\tilde{w}}_i^j}) - \frac{1}{2} \overset{| \Gamma_0 |
      - 1}{\underset{i = 1}{\sum}} \underset{j = 1}{\overset{| \Gamma_0 | -
      1}{\sum}} k_i \cdot (M_{\Gamma} (t)^{- 1})^{i j}
      \bar{k}_j} P (\overline{\tilde{w}}_i)
      \omega \left|_{\tilde{w}_i = 0} \right.,
    \end{eqnarray*}
    and
    \[ {e^{- \frac{1}{2} \overset{| \Gamma_0 | - 1}{\underset{i = 1}{\sum}}
       \underset{j = 1}{\overset{| \Gamma_0 | - 1}{\sum}} k_i
       \cdot (M_{\Gamma} (t)^{- 1})^{i j} \bar{k}_j}} 
       \leqslant 1, \]
    we can get the following inequality after some direct computations:
    \[ \left| D \int_{(\mathbb{C}^d)^{| \Gamma_0 | - 1}} {\Delta'}^{\ast}
       \tilde{W} ((\Gamma, n), \Phi_{k}) \right| \leqslant C'' \left( 1
       + \sum_{i = 1}^{| \Gamma_0 | - 1} | k_i |
       \right)^{\underset{e \in \Gamma_1, 1 \leqslant i \leqslant
       d}{\overset{}{\sum}} n_{i, e} + | \Gamma_1 | + 2 | D |}, \]
    where $C''$ is some constant, $D$ is any order $| D |$ differential
    operator on on $\mathbb{C}^d_{w_{| \Gamma_0 |}} \times \widetilde{[0,
    L]^{| \Gamma_1 |}}$
  \end{enumeratenumeric}
\end{proof}

With above lemma, we can get our main theorem:

\begin{theorem}
  \label{extension theorem1}Given a connected decorated graph $(\Gamma, n)$
  without self-loops, and $\Phi \in \Omega^{\ast, \ast}_c ((\mathbb{C}^d)^{|
  \Gamma_0 |})$, $\int_{(\mathbb{C}^d)^{| \Gamma_0 | - 1}} \tilde{W} ((\Gamma,
  n), \Phi)$ can be extended to a smooth form on $\mathbb{C}^d_{w_{| \Gamma_0
  |}} \times \widetilde{[0, L]^{| \Gamma_1 |}}$. Furthermore, the map
  \[ \Phi \in \Omega^{\ast, \ast}_c ((\mathbb{C}^d)^{| \Gamma_0 |})
     \rightarrow \int_{(\mathbb{C}^d)^{| \Gamma_0 | - 1}} \tilde{W} ((\Gamma,
     n), \Phi) \in \Omega^{\ast} \left( \mathbb{C}^d_{w_{| \Gamma_0 |}} \times
     \widetilde{[0, L]^{| \Gamma_1 |}} \right) \]
  is a continuous map between topological vector spaces.
\end{theorem}

\begin{proof}
  We denote the space of Schwartz differential forms\footnote{A differential
  form on linear space is called Schwartz, if all its derivatives decay faster
  than any polynomials.} on $(\mathbb{C}^d)^{| \Gamma_0 | - 1}$ by $S^{\ast,
  \ast} ((\mathbb{C}^d)^{| \Gamma_0 | - 1})$, then the embeding from
  $\Omega^{\ast, \ast}_c ((\mathbb{C}^d)^{| \Gamma_0 |})$ to $S^{\ast, \ast}
  ((\mathbb{C}^d)^{| \Gamma_0 | - 1})$ is continuous. We define the following
  map from $S^{\ast, \ast} ((\mathbb{C}^d)^{| \Gamma_0 | - 1})$ to
  $\Omega^{\ast} \left( \mathbb{C}^d_{w_{| \Gamma_0 |}} \times \widetilde{[0,
  L]^{| \Gamma_1 |}} \right)$:
  \[ f : \Phi \rightarrow \int_{(\mathbb{C}^d)^{| \Gamma_0 | - 1}_{k}}
     \int_{(\mathbb{C}^d)^{| \Gamma_0 | - 1}} \tilde{W} ((\Gamma, n),
     \Phi_{k}) \tilde{\Phi} (k, \bar{k}) \prod_{i
     = 1}^{| \Gamma_0 | - 1} d^d k_i d^d \bar{k}_i, \]
  where $\tilde{\Phi} (k, \bar{k})$ is the Fourier
  transfrom of $\Phi$. By lemma \ref{extension lemma}, we have
  \begin{equation}
    \left| D \int_{(\mathbb{C}^d)^{| \Gamma_0 | - 1}} {\Delta'}^{\ast}
    \tilde{W} ((\Gamma, n), \Phi_{k}) \right| \leqslant C'' \left( 1 +
    \sum_{i = 1}^{| \Gamma_0 | - 1} | k_i |
    \right)^{\underset{e \in \Gamma_1, 1 \leqslant i \leqslant
    d}{\overset{}{\sum}} n_{i, e} + | \Gamma_1 | + 2 | D |}, \label{order
    estimate}
  \end{equation}
  where $D$ is any order $| D |$ differential operator.
  
  Since Fourier transform is a continuous isomorphism between $S^{\ast, \ast}
  ((\mathbb{C}^d)^{| \Gamma_0 | - 1})$, for any $m \geqslant 0$, we have
  \begin{equation}
    \left| \left( 1 + \sum_{i = 1}^{| \Gamma_0 | - 1} | k_i |
    \right)^{\underset{e \in \Gamma_1, 1 \leqslant i \leqslant
    d}{\overset{}{\sum}} n_{i, e} + | \Gamma_1 | + 2 | D | + m} \tilde{\Phi}
    \right| \leqslant \max \{ | D' \Phi | \},
  \end{equation}
  where $D'$ is some differential operator with polynomial growth
  coefficients. Then we have
  \[ \left| D \int_{(\mathbb{C}^d)^{| \Gamma_0 | - 1}} \tilde{W} ((\Gamma, n),
     \Phi_{k}) \tilde{\Phi} (k, \bar{k}) \right|
     \leqslant C'' \max \{ | D' \Phi | \} \left( 1 + \sum_{i = 1}^{| \Gamma_0
     | - 1} | k_i | \right)^{- m} . \]
  Therefore
  \[ \int_{(\mathbb{C}^d)^{| \Gamma_0 | - 1}} \tilde{W} ((\Gamma, n),
     \Phi_{k}) \tilde{\Phi} (k, \bar{k}) \prod_{i
     = 1}^{| \Gamma_0 | - 1} d^d k_i d^d \bar{k}_i \]
  is integrable for large $m$, and
  \[ \left| \int_{(\mathbb{C}^d)^{| \Gamma_0 | - 1}_{k}}
     \int_{(\mathbb{C}^d)^{| \Gamma_0 | - 1}} \tilde{W} ((\Gamma, n),
     \Phi_{k}) \tilde{\Phi} (k, \bar{k}) \prod_{i
     = 1}^{| \Gamma_0 | - 1} d^d k_i d^d \bar{k}_i \right| \leqslant C''' \max
     \{ | D' \Phi | \}, \]
  where $C'''$ is some constant. This shows $f$ is a continuous map. Finally,
  we notice that $f \left|_{_{\Omega^{\ast, \ast}_c ((\mathbb{C}^d)^{|
  \Gamma_0 |})}} \right.$ is the map
  \[ \Phi \in \Omega^{\ast, \ast}_c ((\mathbb{C}^d)^{| \Gamma_0 |})
     \rightarrow \int_{(\mathbb{C}^d)^{| \Gamma_0 | - 1}} \tilde{W} ((\Gamma,
     n), \Phi) \in \Omega^{\ast} \left( \mathbb{C}^d_{w_{| \Gamma_0 |}} \times
     \widetilde{[0, L]^{| \Gamma_1 |}} \right), \]
  the proof has been completed.
\end{proof}

\begin{corollary}
  \label{extension theorem2}Given a connected decorated graph $(\Gamma, n)$
  without self-loops, and $\Phi \in \Omega^{\ast, \ast}_c ((\mathbb{C}^d)^{|
  \Gamma_0 |})$, $\int_{\mathbb{C}^d_{w_{| \Gamma_0 |}}}
  \int_{(\mathbb{C}^d)^{| \Gamma_0 | - 1}} \tilde{W} ((\Gamma, n), \Phi)$ can
  be extended to a smooth differential form on $\widetilde{[0, L]^{| \Gamma_1
  |}}$. Furthermore, the map
  \[ \Phi \in \Omega^{\ast, \ast}_c ((\mathbb{C}^d)^{| \Gamma_0 |})
     \rightarrow \int_{\mathbb{C}^d_{w_{| \Gamma_0 |}}}
     \int_{(\mathbb{C}^d)^{| \Gamma_0 | - 1}} \tilde{W} ((\Gamma, n), \Phi)
     \in \Omega^{\ast} \left( \widetilde{[0, L]^{| \Gamma_1 |}} \right) \]
  is a continuous map between topological vector spaces.
\end{corollary}

\begin{proof}
  We note the support of $\int_{(\mathbb{C}^d)^{| \Gamma_0 | - 1}} \tilde{W}
  ((\Gamma, n), \Phi)$ is compact, then our claim follows from Theorem
  \ref{extension theorem1} and the following fact:
  
  the integration over $\mathbb{C}^d_{w_{| \Gamma_0 |}}$ is a continuous map
  from $\Omega^{\ast}_c \left( \mathbb{C}^d_{w_{| \Gamma_0 |}} \times
  \widetilde{[0, L]^{| \Gamma_1 |}} \right)$ to $\Omega^{\ast} \left(
  \widetilde{[0, L]^{| \Gamma_1 |}} \right)$.
\end{proof}

\begin{corollary}
  \label{extension theorem3}Given a connected decorated graph $(\Gamma, n)$
  without self-loops, and $\Phi \in \Omega^{\ast, \ast}_c ((\mathbb{C}^d)^{|
  \Gamma_0 |})$, then $W_0^L ((\Gamma, n), \Phi)$ exists. Furthermore,
  the map
  \[ \Phi \in \Omega^{\ast, \ast}_c ((\mathbb{C}^d)^{| \Gamma_0 |})
     \rightarrow W_0^L ((\Gamma, n), \Phi) \]
  defines a distribution on $(\mathbb{C}^d)^{| \Gamma_0 |}$.
\end{corollary}

\begin{proof}
  This follows from Corollary \ref{extension theorem2} and the following fact:
  
  The integration over $\widetilde{[0, L]^{| \Gamma_1 |}}$ is a continuous map
  from $\Omega^{\ast} \left( \widetilde{[0, L]^{| \Gamma_1 |}} \right)$ to
  $\mathbb{C}$.
\end{proof}

\begin{remark}
  If we carefully compute the order of the distribution $W_0^L ((\Gamma, n),
  -)$, we will find out the order is less than or equal to $\underset{e \in \Gamma_1, 1
  \leqslant i \leqslant d}{\overset{}{\sum}} n_{i, e} + | \Gamma_1 |$. We
  leave the proof to the reader. 
\end{remark}

\begin{corollary}
  \label{distribution theorem}Given a connected decorated graph $(\Gamma, n)$
  without self-loops, and $\Phi \in \Omega^{\ast, \ast}_c ((\mathbb{C}^d)^{|
  \Gamma_0 |})$, the map
  \[ \Phi \in \Omega^{\ast, \ast}_c ((\mathbb{C}^d)^{| \Gamma_0 |})
     \rightarrow W_0^{+ \infty} ((\Gamma, n), \Phi) \]
  defines a distribution on $(\mathbb{C}^d)^{| \Gamma_0 |}$. Furthermore,
  \[ W_0^{+ \infty} ((\Gamma, n), -) = \lim_{L \rightarrow + \infty} W_0^L
     ((\Gamma, n), -) \]
  as distributions.
\end{corollary}

\begin{proof}
  By Proposition \ref{finiteness transfer}, we know
  \[
  W_0^{+ \infty} ((\Gamma, n), \Phi)=\lim_{\underset{L \rightarrow +
     \infty}{\varepsilon \rightarrow 0}}W_{\varepsilon}^{L} ((\Gamma, n), \Phi)=\lim_{L\rightarrow+\infty}W_{0}^{L} ((\Gamma, n), \Phi).
  \]
  Our claim follows from the following fact in distribution theory: for a family of distributions $T_{s}\in \mathcal{D}^{\ast, \ast} ((\mathbb{C}^d)^{| \Gamma_0 |})$, if $\lim_{s\rightarrow s_{0}} T_{s}(\phi)$ exists for any $\phi\in\Omega^{\ast, \ast}_c ((\mathbb{C}^d)^{|
  \Gamma_0 |})$, then $\lim_{s\rightarrow s_{0}}T_{s}$ exists as distribution. This can be proved by Banach–Steinhaus theorem for $\Omega^{\ast, \ast}_c
  ((\mathbb{C}^d)^{| \Gamma_0 |})$. The proof is similar to (and much easier than) the proof of Proposition \ref{finiteness transfer}, we left the details to the reader.  
\end{proof}

\begin{remark}
  For topological holomorphic theories, we also have analogous results to
  Theorem \ref{extension theorem1}, see {\cite{minghaoBrian2024}} for details.
\end{remark}

\section{Holomorphicity of Feynman graph integrals.}

\subsection{Holomorphicity and integrals over boundaries of compactified
Schwinger spaces.}

From Corollary \ref{distribution theorem}, we know $W_0^{+ \infty} ((\Gamma,
n), -)$ is a distribution. It is natural to ask whether $W_0^{+ \infty}
((\Gamma, n), -)$ is holomorphic, i.e.
\begin{equation}
  \bar{\partial} W_0^{+ \infty} ((\Gamma, n), -) = 0. \label{holomorphicity}
\end{equation}
Unfortunately, $\left( \ref{holomorphicity} \right)$ is not true in general.
The failure of $\left( \ref{holomorphicity} \right)$ is called
anomalies by physicists. In this subsection, we show these anomalies can be
computed by integrals over boundaries of compactified Schwinger spaces.

\begin{proposition}
  \label{differential transfer}Given a connected decorated graph $(\Gamma, n)$
  without self-loops, and $\Phi \in \Omega^{\ast, \ast}_c ((\mathbb{C}^d)^{|
  \Gamma_0 |})$, $L > 0$, the following equalities hold:
  \begin{enumeratenumeric}
    \item
    \begin{equation}
      (\bar{\partial} W_0^L) ((\Gamma, n), -) = (- 1)^{| \Gamma_1 |}
      \int_{\partial \widetilde{[0, L]^{| \Gamma_1 |}}}
      \int_{(\mathbb{C}^d)^{| \Gamma_0 |}} \tilde{W} ((\Gamma, n), -) .
      \label{QME}
    \end{equation}
    \item
    \[ (\bar{\partial} W_0^{+ \infty}) ((\Gamma, n), -) = \lim_{L \rightarrow
       + \infty} (- 1)^{| \Gamma_1 |} \int_{\partial \widetilde{[0, L]^{|
       \Gamma_1 |}}} \int_{(\mathbb{C}^d)^{| \Gamma_0 |}} \tilde{W} ((\Gamma,
       n), -) . \]
  \end{enumeratenumeric}
\end{proposition}

\begin{proof}
  
  \begin{enumeratenumeric}
    \item We know
    \begin{eqnarray*}
      &  & \tilde{W} ((\Gamma, n), \bar{\partial} \Phi)\\
      & = & (- 1)^{(d - 1) | \Gamma_1 |} \bar{\partial} \left( \prod_{e \in
      \Gamma_1 } \partial_{z_{h (e)}}^{n (e)} P_{t_e} (z_{h
      (e)} - z_{t (e)}, \bar{z}_{h (e)} -
      \bar{z}_{t (e)}) \wedge \Phi \right)\\
      & - & (- 1)^{(d - 1) | \Gamma_1 |} \bar{\partial} \left( \prod_{e \in
      \Gamma_1 } \partial_{z_{h (e)}}^{n (e)} P_{t_e} (z_{h
      (e)} - z_{t (e)}, \bar{z}_{h (e)} -
      \bar{z}_{t (e)}) \right) \wedge \Phi.
    \end{eqnarray*}
    By Lemma \ref{closeness of propagator}, the second term equals
    \begin{eqnarray*}
        & &-(- 1)^{(d - 1) | \Gamma_1 |}\bar{\partial} \left( \prod_{e \in
      \Gamma_1 } \partial_{z_{h (e)}}^{n (e)} P_{t_e} (z_{h
      (e)} - z_{t (e)}, \bar{z}_{h (e)} -
      \bar{z}_{t (e)}) \right) \wedge \Phi\\
      &=&
      (- 1)^{(d - 1) | \Gamma_1 |} d_t \left( \prod_{e \in \Gamma_1 }
      \partial_{z_{h (e)}}^{n (e)} P_{t_e} (z_{h (e)} -
      z_{t (e)}, \bar{z}_{h (e)} -
      \bar{z}_{t (e)}) \right) \wedge \Phi\\
      &=&
      (- 1)^{(d - 1) | \Gamma_1 |} d_t (\tilde{W} ((\Gamma, n), \Phi)).
    \end{eqnarray*}
    Therefore, by Stokes theorem, we have
    \begin{eqnarray*}
      &  & (\bar{\partial} W_0^L) ((\Gamma, n), \Phi)\\
      & = & {(- 1)^{d | \Gamma_1 |}}  W_0^L ((\Gamma, n), \bar{\partial}
      \Phi)\\
      & = & {(- 1)^{d | \Gamma_1 |}}  \int_{\widetilde{[0, L]^{| \Gamma_1
      |}}} \int_{(\mathbb{C}^d)^{| \Gamma_0 |}} \tilde{W} ((\Gamma, n),
      \bar{\partial} \Phi)\\
      & = & (- 1)^{| \Gamma_1 |} \int_{\widetilde{[0, L]^{| \Gamma_1 |}}}
      \int_{(\mathbb{C}^d)^{| \Gamma_0 |}} d_t (\tilde{W} ((\Gamma, n),
      \Phi))\\
      & = & (- 1)^{| \Gamma_1 |} \int_{\partial \widetilde{[0, L]^{| \Gamma_1
      |}}} \int_{(\mathbb{C}^d)^{| \Gamma_0 |}} \tilde{W} ((\Gamma, n), \Phi)
      .
    \end{eqnarray*}
    \item We know taking derivatives commutes with limits of distributions,
    i.e.
    \[ \bar{\partial} W_0^{+ \infty} ((\Gamma, n), -) = \lim_{L + \infty}
       \bar{\partial} W_0^L ((\Gamma, n), -), \]
    then the conclusion follows.
  \end{enumeratenumeric}
\end{proof}

Let's describe the boundaries of compactified Schwinger space $\widetilde{[0,
L]^{| \Gamma_1 |}}$ in detail.

Given $\Gamma'$ is a subgraph of $\Gamma$, then both $\Gamma'_1 = \{ e'_1,
\ldots e'_{| \Gamma'_1 |} \}$ and $\Gamma_1 = \{ e_1, \ldots, e_{| \Gamma_1 |}
\}$ are ordered sets. Assume $\Gamma_1 \backslash \Gamma_1' = \left\{ e_{i_1},
\ldots, e_{i_{| \Gamma_1 | - | \Gamma_1' |}} \right\}$, then there exists a
unique permutation \[
\left\{ e_{i_1}, \ldots, e_{i_{| \Gamma_1 | - | \Gamma_1'
|}}, e'_1, \ldots e'_{| \Gamma'_1 |} \right\} \rightarrow \{ e_1, \ldots, e_{|
\Gamma_1 |} \}.\]
 We denote the sign of this permutation by $(- 1)^{\sigma
(\Gamma', \Gamma / \Gamma')}$, where $\Gamma / \Gamma'$ is the graph obtained
by the quotient of $\Gamma$ by $\Gamma'$.

The boundaries $\partial \widetilde{[0, L]^{| \Gamma_1 |}}$ has the following
decomposition:
\[ \partial \widetilde{[0, L]^{| \Gamma_1 |}} = \left( - \partial_0
   \widetilde{[0, L]^{| \Gamma_1 |}} \right) \cup \partial_L \widetilde{[0,
   L]^{| \Gamma_1 |}}, \]
where $\partial_0 \widetilde{[0, L]^{| \Gamma_1 |}}$($\partial_L
\widetilde{[0, L]^{| \Gamma_1 |}}$) describe the boundary components near the
origin (away from the origin). More precisely, we have
\[ \left\{\begin{array}{l}
     \partial_0 \widetilde{[0, L]^{| \Gamma_1 |}} = \bigcup_{\Gamma' \subseteq
     \Gamma} (- 1)^{\sigma (\Gamma', \Gamma / \Gamma')} \widetilde{(0, +
     \infty)^{| \Gamma_1' |}} /\mathbb{R}^+ \times \widetilde{[0, + L]^{|
     \Gamma_1 \backslash \Gamma_1' |}}\\
     \partial_L \widetilde{[0, L]^{| \Gamma_1 |}} = \bigcup_{e \in \Gamma_1}
     (- 1)^{| e |} \{ L \} \times  \widetilde{[0, L]^{| (\Gamma / e)_1 |}}
   \end{array}\right. . \]
\begin{remark}\label{QME explanation}
  If the reader is familar with Batalin--Vilkovisky
  formalism{\cite{BATALIN198127,batalin1983generalizedbf0904,Schwarz:1992nx,costellorenormalization,1997aksz}},
  it should be obvious that formula $\left( \ref{QME} \right)$ can be used to
  compute the failure of the quantum master equation:
  \[ \left( Q I + \frac{1}{2} \{ I, I \} + \hbar \Delta_{\tmop{BV}} I \right)
     e^{\frac{1}{\hbar} I} = O e^{\frac{1}{\hbar} I} . \]
  More specifically, the term $(\bar{\partial} W_0^L) ((\Gamma, n), -)$
  corresponds to $(Q I) e^{\frac{1}{\hbar} I}$, the term
  \[ \int_{\partial_L \widetilde{[0, L]^{| \Gamma_1 |}}}
     \int_{(\mathbb{C}^d)^{| \Gamma_0 |}} \tilde{W} ((\Gamma, n), -) \]
  corresponds to $\left( \frac{1}{2} \{ I, I \} + \hbar \Delta_{\tmop{BV}} I
  \right) e^{\frac{1}{\hbar} I}$, and the term
  \[ \int_{\partial_0 \widetilde{[0, L]^{| \Gamma_1 |}}}
     \int_{(\mathbb{C}^d)^{| \Gamma_0 |}} \tilde{W} ((\Gamma, n), -) \]
  corresponds to the anomaly $O e^{\frac{1}{\hbar} I}$.
\end{remark}

The following result shows that the integrals over $\partial_L \widetilde{[0,
L]^{| \Gamma_1 |}}$ have no contributions to $(\bar{\partial} W_0^{+ \infty})
((\Gamma, n), -)$:

\begin{proposition}
  \label{no IR boundary integrals}Given a connected decorated graph $(\Gamma,
  n)$ without self-loops, and $\Phi \in \Omega^{\ast, \ast}_c
  ((\mathbb{C}^d)^{| \Gamma_0 |})$, $e \in \Gamma_1$ is any edge in $\Gamma$,
  then
  \[ \lim_{L \rightarrow + \infty} \int_{\{ L \} \times  \widetilde{[0, L]^{|
     (\Gamma / e)_1 |}}} \int_{(\mathbb{C}^d)^{| \Gamma_0 |}} \tilde{W}
     ((\Gamma, n), \Phi) = 0 \]
\end{proposition}

\begin{proof}
  Let $\Gamma \backslash e$ be the subgraph generated by edges in $\Gamma_1
  \backslash \{ e \}$, with same vertices as $\Gamma$, and $\tilde{\Phi} (L)
  \in \Omega^{\ast, \ast}_c ((\mathbb{C}^d)^{| \Gamma_0 |})$ is defined by the
  following formula:
  \[ \tilde{\Phi} (L) = H (t_e, z_{h (e)} -
     z_{t (e)}, \bar{z}_{h (e)} -
     \bar{z}_{t (e)}) \wedge \Phi . \]
  Then we perform the following calculation:
  \begin{eqnarray*}
    &  & \lim_{L \rightarrow + \infty} \int_{\{ L \} \times  \widetilde{[0,
    L]^{| (\Gamma / e)_1 |}}} \int_{(\mathbb{C}^d)^{| \Gamma_0 |}} \tilde{W}
    ((\Gamma, n), \Phi)\\
    & = & (\pm) \lim_{L \rightarrow + \infty} \lim_{\varepsilon \rightarrow
    0} \int_{\{ L \} \times  [\varepsilon, L]^{| (\Gamma / e)_1 |}}
    \int_{(\mathbb{C}^d)^{| \Gamma_0 |}} \\
    & &\prod_{e' \in \Gamma_1 \backslash \{
    e \} } \partial_{z_{h (e')}}^{n (e')} P_{t_{e'}} (z_{h
    (e')} - z_{t (e')}, \bar{z}_{h (e')} -
    \bar{z}_{t (e')})\\
    & \wedge & H (t_e, z_{h (e)} - z_{t
    (e)}, \bar{z}_{h (e)} - \bar{z}_{t
    (e)}) \wedge \Phi\\
    & = & (\pm) \lim_{L \rightarrow + \infty} W_0^L ((\Gamma \backslash e, n
    |  \nobracket_{\Gamma \backslash e}), \tilde{\Phi} (L)) .
  \end{eqnarray*}
  By Banach–Steinhaus theorem for Frechet space $\Omega^{\ast, \ast} (K)$,
  where $K$ is the support of $\Phi$, we have
  \[ | W_0^L ((\Gamma \backslash e, n |  \nobracket_{\Gamma \backslash e}),
     \varphi) | \leqslant C max \{ | D \varphi | \} \quad \text{ for any
     } \varphi \in \Omega^{\ast, \ast} (K), \]
  where $D$ is some differential operator on $K$, $C$ is some constant.
  Therefore
  \begin{eqnarray*}
    &  & \left| \lim_{L \rightarrow + \infty} \int_{\{ L \} \times 
    \widetilde{[0, L]^{| (\Gamma / e)_1 |}}} \int_{(\mathbb{C}^d)^{| \Gamma_0
    |}} \tilde{W} ((\Gamma, n), \Phi) \right|\\
    & = & | \lim_{L \rightarrow + \infty} W_0^L ((\Gamma \backslash e, n | 
    \nobracket_{\Gamma \backslash e}), \tilde{\Phi} (L)) |\\
    & = & \lim_{L \rightarrow + \infty} | W_0^L ((\Gamma \backslash e, n | 
    \nobracket_{\Gamma \backslash e}), \tilde{\Phi} (L)) |\\
    & \leqslant & \lim_{L \rightarrow + \infty} C max \{ | D
    \tilde{\Phi} (L) | \}\\
    & = & 0.
  \end{eqnarray*}
  
\end{proof}

By Proposition \ref{no IR boundary integrals}, we know the anomaly is given by
integrals over $\partial_0 \widetilde{[0, L]^{| \Gamma_1 |}}$.

\subsection{Anomaly integrals and Laman graphs.}\label{Laman graph integral}

In this subsection, we will describe integrals over $\partial_0 \widetilde{[0,
L]^{| \Gamma_1 |}}$ in details. Given a subgraph $\Gamma' \subseteq \Gamma$,
we notice the following equality:
\[ \int_{\widetilde{(0, + \infty)^{| \Gamma_1' |}} /\mathbb{R}^+ \times
   \widetilde{[0, + L]^{| \Gamma_1 \backslash \Gamma_1' |}}}
   \int_{(\mathbb{C}^d)^{| \Gamma_0 |}} \tilde{W} ((\Gamma, n), -) =
   \int_{\partial C_{\Gamma'_1} \cap \widetilde{[0, L]^{| \Gamma_1 |}}}
   \int_{(\mathbb{C}^d)^{| \Gamma_0 |}} \tilde{W} ((\Gamma, n), -) . \]
This equality will provide us a way to perform concrete computations. We first
concentrate on the case $\Gamma' = \Gamma$. In this case, $\partial C_{\Gamma_1}
\cap \widetilde{[0, L]^{| \Gamma_1 |}} = \partial C_{\Gamma_1}$. In the
following, we use $O_{(\Gamma, n)}$ to denote the following integral:
\[ O_{(\Gamma, n)} (\Phi) = \int_{\partial C_{\Gamma_1}} \int_{(\mathbb{C}^d)^{|
   \Gamma_0 |}} \tilde{W} ((\Gamma, n), \Phi), \text{\quad for } \Phi \in \Phi
   \in \Omega^{\ast, \ast}_c ((\mathbb{C}^d)^{| \Gamma_0 |}) . \]
\begin{definition}
  Given a connected directed graph $\Gamma$ without self-loops, we call
  $\Gamma$ a Laman graph, if the following conditions hold:
  \begin{enumeratenumeric}
    \item For any subgraph $\Gamma'$, we have the following inequality:
    \begin{equation}
      d | \Gamma'_0 | \geqslant (d - 1) | \Gamma_1' | + d + 1. \label{Laman
      condition1}
    \end{equation}
    \item The following equality holds:
    \begin{equation}
      d | \Gamma_0 | = (d - 1) | \Gamma_1 | + d + 1. \label{Laman condition2}
    \end{equation}
  \end{enumeratenumeric}
  We use the notation $\tmop{Laman} (\Gamma)$ to denote all the Laman
  subgraphs of $\Gamma$.
\end{definition}

\begin{remark}
  When $d = 2$, the concept of Laman graph originate from Laman's
  characterization of generic rigidity of graphs embeded in two dimensional
  real linear space (see {\cite{Laman1970OnGA}}). Our definition appears in
  {\cite{budzik2023feynman}}.
\end{remark}

\begin{theorem}
  \label{laman anomaly}Given a decorated graph\footnote{We do not assume
  $\Gamma$ is connected here.} $(\Gamma, n)$ without self-loops, and $\Phi \in
  \Omega^{\ast, \ast}_c ((\mathbb{C}^d)^{| \Gamma_0 |})$. The following
  statements hold:
  \begin{enumeratenumeric}
    \item If $\Gamma$ is not a Laman graph, we have
    \[ O_{(\Gamma, n)} (\Phi) = \int_{\partial C_{\Gamma_1}}
       \int_{(\mathbb{C}^d)^{| \Gamma_0 |}} \tilde{W} ((\Gamma, n), \Phi) = 0.
    \]
    \item If $\Gamma$ is a Laman graph,
    \begin{equation}
      O_{(\Gamma, n)} (\Phi) = \int_{\partial C_{\Gamma_1}}
      \int_{(\mathbb{C}^d)^{| \Gamma_0 |}} \tilde{W} ((\Gamma, n), \Phi) =
      \int_{\mathbb{C}^d_{w_{| \Gamma_0 |}}} (D \Phi)
      \left|_{w_i = 0, i \neq | \Gamma_0 |} \right.,
    \end{equation}
    where $D$ is a differential operator with constant coefficients, which
    only involves holomorphic derivatives. The order of $D$ is
    \[ \underset{e \in \Gamma_1, 1 \leqslant i \leqslant d}{\overset{}{\sum}}
       n_{i, e} + | \Gamma_1 | - 1 \]
  \end{enumeratenumeric}
\end{theorem}

\begin{proof}
  Let's assume $\Gamma$ has $m$ connected components $\Gamma^1, \Gamma^2,
  \ldots, \Gamma^m$. If some $\Gamma^i$ does not satisfy $\left( \ref{Laman
  condition1} \right)$, by Proposition \ref{vanishing result1}, we have
  $\tilde{W} ((\Gamma, n), \Phi) = 0$. Let's assume all $\Gamma^i$ satisfy
  $\left( \ref{Laman condition1} \right)$, then we have
  \[ d | \Gamma_0 | \geqslant (d - 1) | \Gamma_1 | + m (d + 1) . \]
  In this case, $M_{\Gamma} (t)$ is not invertible, but we can easily know
  that it has exactly $m - 1$ zero eigenvectors by studying the weighted
  Laplacian of each subgraph. We still use $\det (M_{\Gamma} (t))$ to denote
  the product of its nonzero eigenvalues.
  
  Note $\tilde{W} ((\Gamma, n), \Phi)$ has the following form:
  \[ \frac{1}{\pi^{d | \Gamma_1 |} \left( \underset{e \in \Gamma_1}{\prod} t_e
     \right)^d} e^{- \frac{1}{2} \overset{| \Gamma_0 | - 1}{\underset{i =
     1}{\sum}} \underset{j = 1}{\overset{| \Gamma_0 | - 1}{\sum}}
     w_i \cdot M_{\Gamma} (t)_{i j}
     \bar{w}_j} \sum_{i = 1}^{| \Gamma_1 |} \sum_{S \subseteq
     \Gamma_1, | S | = i} P_S (y_e) \prod_{e \in S} d t_e
     \wedge \Phi, \]
  where $P_S (y_e)$ is a differential form on
  $(\mathbb{C}^d)^{| \Gamma_0 | - 1}$ with anti-holomorphic polynomial
  coefficients, its degree equals to $\underset{e \in \Gamma_1, 1 \leqslant i
  \leqslant d}{\overset{}{\sum}} n_{i, e} + | S |$.
  
  Using functions $\{ \rho \} \cup \{ \xi_e \}_{e \in \Gamma_1}$, we can
  rewrite $\tilde{W} ((\Gamma, n), \Phi)$ as
  \begin{eqnarray*}
    &  & \frac{1}{\pi^{d | \Gamma_1 |} \rho^{d | \Gamma_1 |} \left(
    \underset{e \in \Gamma_1}{\prod} \xi_e \right)^d} e^{- \frac{1}{2 \rho}
    \overset{| \Gamma_0 | - 1}{\underset{i = 1}{\sum}} \underset{j =
    1}{\overset{| \Gamma_0 | - 1}{\sum}} w_i \cdot M_{\Gamma}
    (\xi)_{i j} \bar{w}j} \sum_{i = 1}^{| \Gamma_1 |}
    \sum_{S \subseteq \Gamma_1, | S | = i} \frac{1}{{\rho^{\underset{e \in
    \Gamma_1, 1 \leqslant i \leqslant d}{\overset{}{\sum}} n_{i, e} + | S |}}
    } {P_S}  \left( \frac{\bar{w}_i}{\xi_e} \right)\\
    & \cdot & \left( \rho^{i -} \prod_{e \in S} d \xi_e + \rho^{i - 1} d \rho
    \sum_{e \in S} (- 1)^{| e | - 1} \xi_e \prod_{e' \in S\backslash \{ e \}}
    d \xi_{e'} \right) \wedge \Phi,
  \end{eqnarray*}
  Note $d \rho |_{\partial C_{\Gamma_1}} = 0 \nobracket$, we have
  \begin{eqnarray*}
    &  & \left( \int_{(\mathbb{C}^d)^{| \Gamma_0 |}} \tilde{W} ((\Gamma, n),
    \Phi) \right) |  \nobracket_{\partial C_{\Gamma_1}}\\
    & = & \frac{1}{\pi^{d | \Gamma_1 |} \left( \underset{e \in
    \Gamma_1}{\prod} \xi_e \right)^d} \sum_{i = 1}^{| \Gamma_1 |} \left(
    \rho^{i - d | \Gamma_1 | - | S | - \underset{e \in \Gamma_1, 1 \leqslant i
    \leqslant d}{\overset{}{\sum}} n_{i, e}} \int_{(\mathbb{C}^d)^{| \Gamma_0
    |}} e^{- \frac{1}{2 \rho} \overset{| \Gamma_0 | - 1}{\underset{i =
    1}{\sum}} \underset{j = 1}{\overset{| \Gamma_0 | - 1}{\sum}}
    w_i \cdot M_{\Gamma} (\xi)_{i j}
    \bar{w}_j} \right.\\
    &  & \left. \sum_{S \subseteq \Gamma_1, | S | = i} {P_S}  \left(
    \frac{\bar{w}_i}{\xi_e} \right) \prod_{e \in S} d \xi_e
    \wedge \Phi \right) | \nobracket_{\partial C_{\Gamma_1}} .
  \end{eqnarray*}
  The top form part of $\left( \int_{(\mathbb{C}^d)^{| \Gamma_0 |}} \tilde{W}
  ((\Gamma, n), \Phi) \right) |  \nobracket_{\partial C_{\Gamma_1}} = \left(
  \int_{(\mathbb{C}^d)^{| \Gamma_0 |}} \tilde{W} ((\Gamma, n), \Phi) \right)
  |_{\rho = 0} \nobracket$ is
  \begin{eqnarray*}
    &  & \frac{1}{\pi^{d | \Gamma_1 |} \left( \underset{e \in
    \Gamma_1}{\prod} \xi_e \right)^d} \left( \rho^{- d | \Gamma_1 | -
    \underset{e \in \Gamma_1, 1 \leqslant i \leqslant d}{\overset{}{\sum}}
    n_{i, e}} \int_{(\mathbb{C}^d)^{| \Gamma_0 |}} e^{- \frac{1}{2 \rho}
    \overset{| \Gamma_0 | - 1}{\underset{i = 1}{\sum}} \underset{j =
    1}{\overset{| \Gamma_0 | - 1}{\sum}} w_i \cdot M_{\Gamma}
    (\xi)_{i j} \bar{w}_j} \right.\\
    & \cdot & \left. \sum_{e \in \Gamma_1} {P_{\Gamma_1 \backslash e}} 
    \left( \frac{\bar{w}_i}{\xi_e} \right) \prod_{e' \in
    \Gamma_1 \backslash e} d \xi_{e'} \wedge \Phi \right) | \nobracket_{\rho =
    0} .
  \end{eqnarray*}
  We have the following asymptotic formula for Gaussian type integral (see
  {\cite[Section 2.3]{Duistermaat2011}} for a proof) when $\rho \rightarrow
  0$:
  \begin{eqnarray*}
    &  & \frac{1}{\pi^{d | \Gamma_1 |} \left( \underset{e \in
    \Gamma_1}{\prod} \xi_e \right)^d} \rho^{- d | \Gamma_1 | - \underset{e \in
    \Gamma_1, 1 \leqslant i \leqslant d}{\overset{}{\sum}} n_{i, e}}
    \int_{(\mathbb{C}^d)^{| \Gamma_0 | - 1}} e^{- \frac{1}{2 \rho} \overset{|
    \Gamma_0 | - 1}{\underset{i = 1}{\sum}} \underset{j = 1}{\overset{|
    \Gamma_0 | - 1}{\sum}} w_i \cdot M_{\Gamma} (\xi)_{i j}
    \bar{w}_j}\\
    & \cdot & \sum_{e \in \Gamma_1} {P_{\Gamma_1 \backslash e}}  \left(
    \frac{\bar{w}_i}{\xi_e} \right) \prod_{e' \in \Gamma_1
    \backslash e} d \xi_{e'} \wedge \Phi\\
    & \sim & \frac{(2 \pi)^{d (| \Gamma_0 | - m)}}{\pi^{d | \Gamma_1 |}
    \left( \underset{e \in \Gamma_1}{\prod} \xi_e \right)^d} \left(
    \frac{\rho^{d | \Gamma_0 | - d | \Gamma_1 | - m d - \underset{e \in
    \Gamma_1, 1 \leqslant i \leqslant d}{\overset{}{\sum}} n_{i, e}}}{\det
    (M_{\Gamma} (\xi))^d} e^{\frac{1}{2} \rho \overset{| \Gamma_0 | -
    1}{\underset{i = 1}{\sum}} \underset{j = 1}{\overset{| \Gamma_0 | -
    1}{\sum}} \overset{d}{\underset{k = 1}{\sum}} M_{\Gamma} (\xi)_{i j}
    \partial_{w_i^k} \partial_{\tilde{w}^k_j}} \right. \iota_{\omega^{\vee}}\\
    & \cdot & \left. \sum_{e \in \Gamma_1} {P_{\Gamma_1 \backslash e}} 
    \left( \frac{\bar{w}_i}{\xi_e} \right) \prod_{e' \in
    \Gamma_1 \backslash e} d \xi_{e'} \wedge \Phi \right)
    \left|_{w_i = 0, i \neq | \Gamma_0 |} \right.\\
    & = & \frac{(2 \pi)^{d (| \Gamma_0 | - m)}}{\pi^{d | \Gamma_1 |} \left(
    \underset{e \in \Gamma_1}{\prod} \xi_e \right)^d} \sum_{l = \underset{e
    \in \Gamma_1, 1 \leqslant i \leqslant d}{\overset{}{\sum}} n_{i, e} + |
    \Gamma_1 | - 1}^{+ \infty} \frac{\rho^{l + d | \Gamma_0 | - d | \Gamma_1 |
    - m d - \underset{e \in \Gamma_1, 1 \leqslant i \leqslant
    d}{\overset{}{\sum}} n_{i, e}}}{\det (M_{\Gamma} (\xi))^d}
    \iota_{\omega^{\vee}}\\
    &  & \left. \frac{1}{l!} \left( \left( \frac{1}{2} \overset{| \Gamma_0 |
    - 1}{\underset{i = 1}{\sum}} \underset{j = 1}{\overset{| \Gamma_0 | -
    1}{\sum}} \overset{d}{\underset{k = 1}{\sum}} M_{\Gamma} (\xi)_{i j}
    \partial_{w_i^k} \partial_{\tilde{w}^k_j} \right)^l \right. \sum_{e \in
    \Gamma_1} {P_{\Gamma_1 \backslash e}}  \left(
    \frac{\bar{w}_i}{\xi_e} \right) \prod_{e' \in \Gamma_1
    \backslash e} d \xi_{e'} \wedge \Phi \right) \left|_{w_i
    = 0, i \neq | \Gamma_0 |} \right.,
  \end{eqnarray*}
  where $\iota_{\omega^{\vee}}$ is the interior product of a differential form with a
  nonzero constant top polyvector field $\omega^{\vee}$ on $(\mathbb{C}^d)^{|
  \Gamma_0 | - 1}$, i.e. it maps a volume form to its coefficient.
  
  We notice that the power of $\rho$ in the asymptotic formula is at least
  \[ d | \Gamma_0 | - (d - 1) | \Gamma_1 | - m d - 1 \geqslant d | \Gamma_0 |
     - (d - 1) | \Gamma_1 | - m (d + 1), \]
  then we know if $\Gamma$ is not a (connected) Laman graph,
  \[ \int_{\partial C_{\Gamma_1}} \int_{(\mathbb{C}^d)^{| \Gamma_0 |}} \tilde{W}
     ((\Gamma, n), \Phi) = 0. \]
  When $\Gamma$ is indeed a Laman graph, we have
  \begin{eqnarray*}
    &  & \int_{\partial C_{\Gamma_1}} \int_{(\mathbb{C}^d)^{| \Gamma_0 |}}
    \tilde{W} ((\Gamma, n), \Phi)\\
    & = & \frac{(2 \pi)^{d (| \Gamma_0 | - 1)}}{\pi^{d | \Gamma_1 |}}
    \int_{\mathbb{C}^d_{w_{| \Gamma_0 |}}} \int_{\partial C_{\Gamma_1}}
    \frac{1}{\left( \underset{e \in \Gamma_1}{\prod} \xi_e \right)^d \det
    (M_{\Gamma} (\xi))^d \left( \underset{e \in \Gamma_1, 1 \leqslant i
    \leqslant d}{\overset{}{\sum}} n_{i, e} + | \Gamma_1 | - 1 \right) !}\\
    & \cdot & \left( \left( \frac{1}{2} \overset{| \Gamma_0 | -
    1}{\underset{i = 1}{\sum}} \underset{j = 1}{\overset{| \Gamma_0 | -
    1}{\sum}} \overset{d}{\underset{k = 1}{\sum}} M_{\Gamma} (\xi)_{i j}
    \partial_{w_i^k} \partial_{\tilde{w}^k_j} \right)^{\underset{e \in
    \Gamma_1, 1 \leqslant i \leqslant d}{\overset{}{\sum}} n_{i, e} + |
    \Gamma_1 | - 1} \right. \iota_{\omega^{\vee}}\\
    & \cdot & \left. \sum_{e \in \Gamma_1} {P_{\Gamma_1 \backslash e}} 
    \left( \frac{\bar{w}_i}{\xi_e} \right) \prod_{e' \in
    \Gamma_1 \backslash e} d \xi_{e'} \wedge \Phi \right)
    \left|_{w_i = 0, i \neq | \Gamma_0 |} \right.,
  \end{eqnarray*}
  Then we can see
  \[ \int_{\partial C_{\Gamma_1}} \int_{(\mathbb{C}^d)^{| \Gamma_0 |}}
     \tilde{W} ((\Gamma, n), \Phi) = \int_{\mathbb{C}^d_{w_{| \Gamma_0 |}}} (D
     \Phi) \left|_{w_i = 0, i \neq | \Gamma_0 |} \right., \]
  where
  \begin{eqnarray*}
    &  & \\
    D & = & \frac{(2 \pi)^{d (| \Gamma_0 | - 1)}}{\pi^{d | \Gamma_1 |}}
    \int_{\partial C_{\Gamma_1}} \frac{1}{\left( \underset{e \in
    \Gamma_1}{\prod} \xi_e \right)^d \det (M_{\Gamma} (\xi))^d \left(
    \underset{e \in \Gamma_1, 1 \leqslant i \leqslant d}{\overset{}{\sum}}
    n_{i, e} + | \Gamma_1 | - 1 \right) !}\\
    & \cdot & \left( \left( \frac{1}{2} \overset{| \Gamma_0 | -
    1}{\underset{i = 1}{\sum}} \underset{j = 1}{\overset{| \Gamma_0 | -
    1}{\sum}} \overset{d}{\underset{k = 1}{\sum}} M_{\Gamma} (\xi)_{i j}
    \partial_{w_i^k} \partial_{\tilde{w}^k_j} \right)^{\underset{e \in
    \Gamma_1, 1 \leqslant i \leqslant d}{\overset{}{\sum}} n_{i, e} + |
    \Gamma_1 | - 1} \right. \iota_{\omega^{\vee}}\\
    & \cdot & \left. \sum_{e \in \Gamma_1} {P_{\Gamma_1 \backslash e}} 
    \left( \frac{\bar{w}_i}{\xi_e} \right) \prod_{e' \in
    \Gamma_1 \backslash e} d \xi_{e'} \right) .
  \end{eqnarray*}
  
\end{proof}

Now, we consider the general case. Given a subgraph $\Gamma' \subseteq
\Gamma$, we use $\Gamma \backslash \Gamma'$ to denote the subgraph with the
vertices $\Gamma_0$, and the edges $\Gamma_1 \backslash
\Gamma'_1$\footnote{Note $\Gamma \backslash \Gamma'$ is different from $\Gamma
/ \Gamma'$.}. We have the following result:

\begin{proposition}
  Given a connected decorated graph $(\Gamma, n)$ without self-loops, \ $\Phi
  \in \Omega^{\ast, \ast}_c ((\mathbb{C}^d)^{| \Gamma_0 |})$, and a subgraph
  $\Gamma'$, then the integral
  \[ \int_{\widetilde{(0, + \infty)^{| \Gamma_1' |}} /\mathbb{R}^+ \times
     \widetilde{[0, + L]^{| \Gamma_1 \backslash \Gamma_1' |}}}
     \int_{(\mathbb{C}^d)^{| \Gamma_0 |}} \tilde{W} ((\Gamma, n), \Phi) =
     \sum_{n' \in \tmop{dec} (\Gamma / \Gamma')} C_{n'} W_0^L ((\Gamma /
     \Gamma', n'), D_{n'} \Phi), \]
  where $\tmop{dec} (\Gamma / \Gamma')$ is the set of all decorations of
  $\Gamma / \Gamma'$, and $C_{n'} = 0$ for all but finite many $n' \in
  \tmop{dec} (\Gamma / \Gamma')$, $D_{n'}$ are differential operators with constant coefficients, which only
  involve holomorphic derivatives.
  
  In particular, when $L \rightarrow +
  \infty$,
  \[ \lim_{L \rightarrow + \infty} \int_{\widetilde{(0, + \infty)^{|
     \Gamma_1' |}} /\mathbb{R}^+ \times \widetilde{[0, + L]^{| \Gamma_1
     \backslash \Gamma_1' |}}} \int_{(\mathbb{C}^d)^{| \Gamma_0 |}} \tilde{W}
     ((\Gamma, n), -) \]
  exists as a distribution.
\end{proposition}

\begin{proof}
  From Theorem \ref{laman anomaly}, we have
  \begin{eqnarray*}
    &  & \int_{\widetilde{(0, + \infty)^{| \Gamma_1' |}} /\mathbb{R}^+ \times
    \widetilde{[0, + L]^{| \Gamma_1 \backslash \Gamma_1' |}}}
    \int_{(\mathbb{C}^d)^{| \Gamma_0 |}} \tilde{W} ((\Gamma, n), \Phi)\\
    & = & {(- 1)^{d \sigma (\Gamma', \Gamma / \Gamma')}} 
    \int_{\widetilde{(0, + \infty)^{| \Gamma_1' |}} /\mathbb{R}^+ \times
    \widetilde{[0, + L]^{| \Gamma_1 \backslash \Gamma_1' |}}}
    \int_{(\mathbb{C}^d)^{| \Gamma_0 |}} \tilde{W} ((\Gamma', n |_{\Gamma'}
    \nobracket), \tilde{W} ((\Gamma \backslash \Gamma', n |_{\Gamma \backslash
    \Gamma'} \nobracket), \Phi))\\
    & = & (- 1)^{(| \Gamma_1' | - 1) (| \Gamma_1 | - | \Gamma'_1 |) + d
    \sigma (\Gamma', \Gamma / \Gamma')} \int_{\widetilde{[0, + L]^{| \Gamma_1
    \backslash \Gamma_1' |}}} \int_{(\mathbb{C}^d)^{| (\Gamma / \Gamma')_0 | -
    1}} \int_{\widetilde{(0, + \infty)^{| \Gamma_1' |}} /\mathbb{R}^+}
    \int_{(\mathbb{C}^d)^{| \Gamma'_0 |}}\\
    &  & \tilde{W} ((\Gamma', n |_{\Gamma'} \nobracket), \tilde{W} ((\Gamma
    \backslash \Gamma', n |_{\Gamma \backslash \Gamma'} \nobracket), \Phi))\\
    & = & (- 1)^{(| \Gamma_1' | - 1) (| \Gamma_1 | - | \Gamma'_1 |) + d
    \sigma (\Gamma', \Gamma / \Gamma')} \int_{\widetilde{[0, + L]^{| \Gamma_1
    \backslash \Gamma_1' |}}} \int_{(\mathbb{C}^d)^{| (\Gamma / \Gamma')_0 | -
    1}}\\
    &  & O_{(\Gamma', n |_{\Gamma'} \nobracket)} (\tilde{W} ((\Gamma
    \backslash \Gamma', n |_{\Gamma \backslash \Gamma'} \nobracket), \Phi))\\
    & = & (- 1)^{(| \Gamma_1' | - 1) (| \Gamma_1 | - | \Gamma'_1 |) + d
    \sigma (\Gamma', \Gamma / \Gamma')} \int_{\widetilde{[0, + L]^{| \Gamma_1
    \backslash \Gamma_1' |}}} \int_{(\mathbb{C}^d)^{| (\Gamma / \Gamma')_0 | -
    1}}\\
    &  & (D \tilde{W} ((\Gamma \backslash \Gamma', n |_{\Gamma \backslash
    \Gamma'} \nobracket), \Phi)) \left|_{w_i = 0, i \neq |
    \Gamma_0 |} \right.,
  \end{eqnarray*}
  where $D$ is a differential operator with constant coefficients, which only
  involves holomorphic derivatives. Therefore, by Leibniz' rule, it is easy to
  see that
  \[ \int_{\widetilde{(0, + \infty)^{| \Gamma_1' |}} /\mathbb{R}^+ \times
     \widetilde{[0, + L]^{| \Gamma_1 \backslash \Gamma_1' |}}}
     \int_{(\mathbb{C}^d)^{| \Gamma_0 |}} \tilde{W} ((\Gamma, n), \Phi) =
     \sum_{n' \in \tmop{dec} (\Gamma / \Gamma')} C_{n'} W_0^L ((\Gamma /
     \Gamma', n'), D_{n'} \Phi), \]
  where $C_{n'} = 0$ for all but finite many $n' \in \tmop{dec} (\Gamma /
  \Gamma')$, $D_{n'}$ are differential operators with constant coefficients, which only
  involve holomorphic derivatives.
\end{proof}

The following corollary immediately follows:

\begin{corollary}
  \label{holomorphicity failure}Given a connected decorated graph $(\Gamma,
  n)$ without self-loops, and $\Phi \in \Omega^{\ast, \ast}_c
  ((\mathbb{C}^d)^{| \Gamma_0 |})$, we have
  \begin{eqnarray*}
    &  & (\bar{\partial} W_0^{+ \infty}) ((\Gamma, n), \Phi)\\
    & = & (- 1)^{| \Gamma_1 | - 1} \sum_{\Gamma' \in \tmop{Laman} (\Gamma)}
    (- 1)^{\sigma (\Gamma', \Gamma / \Gamma')} \int_{\widetilde{(0, +
    \infty)^{| \Gamma_1' |}} /\mathbb{R}^+ \times \widetilde{[0, + L]^{|
    \Gamma_1 \backslash \Gamma_1' |}}} \int_{(\mathbb{C}^d)^{| \Gamma_0 |}}
    \tilde{W} ((\Gamma, n), \Phi) .
  \end{eqnarray*}
\end{corollary}

\subsection{Quadratic relations.}

In this section, we provide a proof of the quadratic relations of $O_{(\Gamma,
n)}$ appeared in {\cite{budzik2023feynman}}. The idea is using Stokes theorem
over $\widetilde{(0, + \infty)^{| \Gamma_1 |}} /\mathbb{R}^+$.

First, from the proof of Proposition \ref{differential transfer}, we have the
following lemma:

\begin{lemma}
  Given a connected decorated graph $(\Gamma, n)$ without self-loops, and
  $\Phi \in \Omega^{\ast, \ast}_c ((\mathbb{C}^d)^{| \Gamma_0 |})$, the
  following equality hold:
  \[ \int_{(\mathbb{C}^d)^{| \Gamma_0 |}} \tilde{W} ((\Gamma, n),
     \bar{\partial} \Phi) = (- 1)^{(d - 1) | \Gamma_1 |} d_t
     \int_{(\mathbb{C}^d)^{| \Gamma_0 |}} \tilde{W} ((\Gamma, n), \Phi) . \]
\end{lemma}

If we perform integration over $\widetilde{(0, + \infty)^{| \Gamma_1 |}}
/\mathbb{R}^+$ on both sides, we will get
\begin{eqnarray}
  O_{(\Gamma, n)} (\bar{\partial} \Phi) & = & \int_{\widetilde{(0, +
  \infty)^{| \Gamma_1 |}} /\mathbb{R}^+} \int_{(\mathbb{C}^d)^{| \Gamma_0 |}}
  \tilde{W} ((\Gamma, n), \bar{\partial} \Phi) \nonumber\\
  & = & (- 1)^{| \Gamma_1 |} \int_{\widetilde{(0, + \infty)^{| \Gamma_1 |}}
  /\mathbb{R}^+} d_t \int_{(\mathbb{C}^d)^{| \Gamma_0 |}} \tilde{W} ((\Gamma,
  n), \Phi) \nonumber\\
  & = & (- 1)^{| \Gamma_1 |} \int_{\partial \widetilde{(0, + \infty)^{|
  \Gamma_1 |}} /\mathbb{R}^+} \int_{(\mathbb{C}^d)^{| \Gamma_0 |}} \tilde{W}
  ((\Gamma, n), \Phi) .  \label{prequatratic relation}
\end{eqnarray}

Before studying the integral over $\partial \widetilde{(0, + \infty)^{|
\Gamma_1 |}} /\mathbb{R}^+$ in detail, we introduce the following definition:

\begin{definition}
  Given a connected decorated graph $(\Gamma, n)$ without self-loops, \ $\Phi
  \in \Omega^{\ast, \ast}_c ((\mathbb{C}^d)^{| \Gamma_0 |})$, and a subgraph
  $\Gamma'$, we define 
\end{definition}

We have the following decomposition for $\partial \widetilde{(0, + \infty)^{|
\Gamma_1 |}} /\mathbb{R}^+$:
\[ \partial \widetilde{(0, + \infty)^{| \Gamma_1 |}} /\mathbb{R}^+ =
   \bigcup_{\Gamma' \subseteq \Gamma} (- 1)^{\sigma (\Gamma', \Gamma /
   \Gamma') + | \Gamma_1' | - 1} \widetilde{(0, + \infty)^{| \Gamma_1' |}}
   /\mathbb{R}^+ \times \widetilde{(0, + \infty)^{| \Gamma_1 \backslash
   \Gamma_1' |}} /\mathbb{R}^+ . \]
Therefore, formula $\left( \ref{prequatratic relation} \right)$ becomes:
\begin{eqnarray*}
  &  & O_{(\Gamma, n)} (\bar{\partial} \Phi)\\
  & = & \sum_{\Gamma' \subseteq \Gamma} (- 1)^{\sigma (\Gamma', \Gamma /
  \Gamma') + | \Gamma_1' | - 1 + (d - 1) | \Gamma_1 |} \int_{\widetilde{(0, +
  \infty)^{| \Gamma_1' |}} /\mathbb{R}^+} \int_{\widetilde{(0, + \infty)^{|
  \Gamma_1 \backslash \Gamma_1' |}} /\mathbb{R}^+} \int_{(\mathbb{C}^d)^{|
  \Gamma_0 |}} \tilde{W} ((\Gamma, n), \Phi)\\
  & = & \sum_{\Gamma' \subseteq \Gamma} (- 1)^{(d + 1) \sigma (\Gamma',
  \Gamma / \Gamma') + | \Gamma_1' | - 1 + (d - 1) | \Gamma_1 |}
  \int_{\widetilde{(0, + \infty)^{| \Gamma_1' |}} /\mathbb{R}^+}
  \int_{\widetilde{(0, + \infty)^{| \Gamma_1 \backslash \Gamma_1' |}}
  /\mathbb{R}^+} \int_{(\mathbb{C}^d)^{| \Gamma_0 |}}\\
  &  & \tilde{W} ((\Gamma', n |_{\Gamma'} \nobracket), \tilde{W} ((\Gamma
  \backslash \Gamma', n |_{\Gamma \backslash \Gamma'} \nobracket), \Phi))\\
  & = & \sum_{\Gamma' \in \tmop{Laman} (\Gamma)} (- 1)^{(d + 1) \sigma
  (\Gamma', \Gamma / \Gamma') + (d - 1) | \Gamma_1 | + (| \Gamma_1' | - 1) (|
  \Gamma_1 | - | \Gamma_1' |)} \int_{\widetilde{(0, + \infty)^{| \Gamma_1
  \backslash \Gamma_1' |}} /\mathbb{R}^+} \int_{(\mathbb{C}^d)^{| (\Gamma /
  \Gamma')_0 | - 1}}\\
  &  & O_{(\Gamma', n |_{\Gamma'} \nobracket)} (\tilde{W} ((\Gamma \backslash
  \Gamma', n |_{\Gamma \backslash \Gamma'} \nobracket), \Phi)) .
\end{eqnarray*}
If we define
\begin{eqnarray*}
  &  & O_{(\Gamma', n |_{\Gamma'} \nobracket)} \circ O_{(\Gamma \backslash
  \Gamma', n |_{\Gamma \backslash \Gamma'} \nobracket)} (\Phi)\\
  & = & (- 1)^{(| \Gamma_1' | - 1) (| \Gamma_1 | - | \Gamma_1' |)}
  \int_{\widetilde{(0, + \infty)^{| \Gamma_1 \backslash \Gamma_1' |}}
  /\mathbb{R}^+} \int_{(\mathbb{C}^d)^{| (\Gamma / \Gamma')_0 | - 1}}\\
  &  & O_{(\Gamma', n |_{\Gamma'} \nobracket)} (\tilde{W} ((\Gamma \backslash
  \Gamma', n |_{\Gamma \backslash \Gamma'} \nobracket), \Phi)),
\end{eqnarray*}
then we have
\begin{eqnarray}
  &  & (- 1)^{(d - 1) | \Gamma_1 |} O_{(\Gamma, n)} (\bar{\partial} \Phi)
  \nonumber\\
  & = & \sum_{\Gamma' \in \tmop{Laman} (\Gamma)} (- 1)^{(d + 1) \sigma
  (\Gamma', \Gamma / \Gamma')} O_{(\Gamma', n |_{\Gamma'} \nobracket)} \circ
  O_{(\Gamma \backslash \Gamma', n |_{\Gamma \backslash \Gamma'} \nobracket)}
  (\Phi) .  \label{prequadratic relation2}
\end{eqnarray}
Finally, we can state the quadratic relations:

\begin{theorem}[K.Budzik, D.Gaiotto, J.Kulp, J.Wu, M.Yu, see
{\cite{budzik2023feynman}}]
  \label{Quadratic relations}Given a connected decorated graph $(\Gamma, n)$
  without self-loops, \ $\Phi \in \Omega^{\ast, \ast}_c ((\mathbb{C}^d)^{|
  \Gamma_0 |})$, we have the following quadratic relations:
  \begin{equation}
    \sum_{\Gamma' \in \tmop{Laman} (\Gamma)} (- 1)^{(d + 1) \sigma (\Gamma',
    \Gamma / \Gamma')} O_{(\Gamma', n |_{\Gamma'} \nobracket)} \circ
    O_{(\Gamma \backslash \Gamma', n |_{\Gamma \backslash \Gamma'}
    \nobracket)} (\Phi) = 0 \label{quadratic relations} .
  \end{equation}
  \begin{proof}
    By Theorem \ref{laman anomaly}, if $\Gamma$ is not a Laman graph,
    $O_{(\Gamma, n)} = 0$. If $\Gamma$ is a Laman graph, we have
    \[ O_{(\Gamma, n)} (\bar{\partial} \Phi) = \int_{\mathbb{C}^d_{w_{|
       \Gamma_0 |}}} (D \bar{\partial} \Phi) \left|_{w_i = 0,
       i \neq | \Gamma_0 |} \right. = \int_{\mathbb{C}^d_{w_{| \Gamma_0 |}}}
       \bar{\partial} (D \Phi) \left|_{w_i = 0, i \neq |
       \Gamma_0 |} \right. = 0. \]
    The theorem follows from $\left( \ref{prequadratic relation2} \right)$.
    
    \ 
  \end{proof}
\end{theorem}

\begin{remark}
  In holomorphic field theories, $O_{(\Gamma, n)}$ should corresponds to
  higher brackets of an $L_{\infty}$ algebra. The quadratic relations $\left(
  \ref{quadratic relations} \right)$ should correspond to the $L_{\infty}$
  relations. A detailed elaboration of this idea for one complex dimensional holomorphic field theories
  can be found in {\cite{Li:2016gcb}}.
\end{remark}

{\nocite{chen2023kontsevich,costello_gwilliam_2016,costello_gwilliam_2021,guisili2021elliptic,rglzditglrmsrfcmdlfmlszj2021,Costello:2012cy,Costello:2019jsy,cmp/1104178138,li2023regularized,Budzik:2023xbr}}

Department of Mathematics \& Statistics, Boston University, Boston, 02215, US.

Email: minghaow@bu.edu

\end{document}